\documentclass[lettersize,journal]{IEEEtran}

\IEEEoverridecommandlockouts

\usepackage[linesnumbered,ruled,noend]{algorithm2e}

\usepackage{subcaption} 
\usepackage{caption}
\usepackage{float}
\usepackage{enumitem}
\usepackage{balance}
\usepackage{cite}
\usepackage{amsmath,amssymb,amsfonts}
\usepackage{bbm}
\usepackage{bm}
\usepackage{url}
\usepackage{amsthm}
\usepackage{graphicx}
\usepackage{xcolor}
\usepackage{stfloats}
\usepackage{amsthm}
\usepackage{amsmath}
\DeclareMathOperator*{\argmin}{arg\,min}
\makeatletter%
\if@twocolumn%
\newcommand{\Figwidth}{\columnwidth}%
\else
\newcommand{\Figwidth}{4.5in}%
\fi%
\makeatother%

\allowdisplaybreaks

\renewcommand{\Figwidth}{3.5in}

\def\log{{\rm{log}}}
\def\min{{\rm{min}}}

\def\bv{\mathbf{v}}

\def\bh{\mathbf{h}}

\def\bs{\mathbf{s}}

\def\by{\mathbf{y}}

\def\bx{\mathbf{x}}

\def\bu{\mathbf{u}}
\def\bI{\mathbf{I}}
\def\be{\mathbf{e}}

\def\bA{\mathbf{A}}

\def\bz{\mathbf{z}}
\def\bA{\mathbf{A}}

\def\bzero{\mathbf{0}}

\def\bxi{\boldsymbol{\xi}}
\def\bepsilon{\boldsymbol{\epsilon}}

\def\btheta{\boldsymbol{\theta}}
\def\bzeta{\boldsymbol{\zeta}}

\def\bgamma{\boldsymbol{\gamma}}

\def\balpha{\boldsymbol{\alpha}}

\newtheorem{remark}{Remark}
\newtheorem{theorem}{Theorem}
\newtheorem{lemma}{Lemma}

\DeclareFontFamily{U}{mathx}{\hyphenchar\font45}
\DeclareFontShape{U}{mathx}{m}{n}{
      <5> <6> <7> <8> <9> <10>
      <10.95> <12> <14.4> <17.28> <20.74> <24.88>
      mathx10
      }{}
\DeclareSymbolFont{mathx}{U}{mathx}{m}{n}
\DeclareFontSubstitution{U}{mathx}{m}{n}
\DeclareMathAccent{\widecheck}{0}{mathx}{"71}
\DeclareMathAccent{\wideparen}{0}{mathx}{"75}

\usepackage{balance}

\begin{document}


\title{Learning OTA: A Unified Framework for \\ Edge Sensing, Computation, and Communication}

\author{Mehdi~Karbalayghareh,~\IEEEmembership{Member,~IEEE,} David~J.~Love,~\IEEEmembership{Fellow,~IEEE}, \\ and~Christopher~G.~Brinton,~\IEEEmembership{Senior Member,~IEEE}

\thanks{M. Karbalayghareh, D. J. Love, and C. G. Brinton are with the Department of Electrical and Computer Engineering, Purdue University, West Lafayette, IN, USA (e-mails: mkarbala@purdue.edu; djlove@purdue.edu; cgb@purdue.edu). This work was supported in part by the Office of Naval Research (ONR) under grant N00014-21-1-2472 and by the National Science Foundation (NSF) under grants CNS2212565, CNS2225578, and CPS2313109.} 

\thanks{A preliminary version of this work has been accepted to the 2026 IEEE Information Theory Workshop (ITW)~\cite{karbalayghareh-ITW2026-ImperfectSensing}, focusing only on sensing-parameter optimization. This journal version extends to a unified joint sensing--computation--communication design with OTA aggregation and convergence-aware resource optimization under link impairments.}}

\maketitle

\begin{abstract}
Most federated learning (FL) studies implicitly assume that training data is readily available at participating edge devices. In practice, this data must first be acquired through heterogeneous sensors with potentially different measurement modalities, and the quality and cost of this acquisition process can directly affect the learning performance. This dependence creates a coupled relationship among sensing, computation, and communication, necessitating a system model that captures their joint effect on resource efficiency and learning accuracy. This paper develops such a unified treatment for an over-the-air (OTA) FL system in which each device configures its sensing modality, power, resolution, and sample size, with aggregations conducted OTA at the server. We derive models for sensing noise, energy consumption, latency, and OTA aggregation distortion, and establish a non-convex convergence bound quantifying how these factors influence the global model trajectory.
Guided by this analysis, we formulate a per-round optimization that balances energy, latency, and learning quality under physical device constraints. The resulting mixed optimization is addressed through an alternating procedure with closed-form and one-dimensional updates for continuous variables and a lightweight discrete search. Numerical experiments corroborate the advantage of jointly optimizing sensing, computation, and communication decisions.
\end{abstract}


\IEEEpeerreviewmaketitle
\section{Introduction}
\label{sec:introduction}
AI-native sixth generation (6G) wireless networks are expected to rely heavily on distributed edge learning, where heterogeneous devices collaboratively train models under latency, reliability, bandwidth, and energy constraints~\cite{Letaief-CommMag2019-6G,Saad-6G-2020,Tataria-6G-2021,Giordani-6G-ComMag2020,Andrews-6G,Brinton-6G-ComMag2025}. Federated learning (FL) provides a scalable framework for this setting by enabling devices to compute local updates and exchange model information rather than raw data~\cite{McMahan2017FedAvg,Li2020FLSurvey,Kairouz2021AdvancesFL,Lim_ComSurvey2020,Nguyen-CommSurvey2021-FL, Latif-CommSurvey2021-FL, Nguyen_JSAC2021,Mehdi_jsait26_coherence}. 

Prior work on wireless FL has largely focused on communication and computation bottlenecks. Communication-efficient FL schemes reduce downlink/uplink overhead through compression, quantization, sparsification, and partial device participation~\cite{Alistarh2017QSGD,Stich2018SparsifiedSGD,Sattler2019SparseFL,Mehdi_infocom26_coherence}. In wireless settings, over-the-air federated learning (OTA-FL) leverages the superposition property of multiple access channels (MAC) to aggregate updates in the uplink (at the server), significantly reducing latency and spectral usage~\cite{Amiri2020ADSGD,Yang2020AirCompFL,Sery2021OTAHetero,Zhu2021OneBitOTA}. Complementary studies incorporate receiver beamforming, power control, and convergence-driven resource allocation to mitigate aggregation distortion~\cite{Cao2022JSACPowerControl}. In parallel, computation-aware FL designs adapt local processing frequency and workload allocation to balance energy consumption and training latency~\cite{Wang2019AdaptiveFL,Nishio2019ClientSelection}. Collectively, these works establish a mature understanding of how communication and computation resources shape learning performance.

Despite this progress, most FL formulations rely on an implicit but critical assumption: training data at each device is already available and of fixed statistical quality. Under this abstraction, the learning pipeline begins with local model updates and ends with global aggregation, while the process by which data is acquired is entirely omitted. In practice, however, training data is generated through \emph{physical sensing} processes that consume energy, incur delay, and introduce measurement distortion. As a result, learning performance is fundamentally constrained not only by how efficiently model updates are communicated and computed, but also by how data is sensed and generated at the devices.

In many edge learning applications, devices are connected to heterogeneous sensors such as cameras, microphones, radar units, or LiDAR that acquire observations from the physical environment. The statistical quality of the sensed data depends on several factors, including sensing modality, sensing resolution, and sensing power. Different sensing modalities capture distinct physical phenomena and produce data with inherently different noise characteristics and feature distributions. Resolution governs the granularity of the acquired measurements and measurement precision, while sensing power determines the effective signal-to-noise ratio (SNR) of the acquired samples. These sensing decisions directly affect the fidelity of the training data and, in turn, influence learning performance and final model accuracy. Consequently, sensing, computation, and communication form \emph{a tightly coupled pipeline} in practical FL systems; isolating any single stage fails to capture the true system-level tradeoffs.

\subsection{Related Works}
Recent works have started to incorporate data acquisition and sensing into federated learning. Liu \emph{et al.}~\cite{Liu2023AmbientIntelligence} studied task-oriented sensing, computation, and communication integration for federated edge learning, where sensing time, energy consumption, and communication resources are jointly considered. Liang \emph{et al.}~\cite{Liang2025ISCC} introduced an integrated sensing--communication--computation (ISCC) formulation for OTA-FL, showing that the number of sensed samples, computation resources, and uplink transmission design should be jointly optimized when data are acquired during training. Du \emph{et al.}~\cite{Du2024ISCCOTAFL} further considered ISCC-enabled OTA-FL in wireless networks and analyzed the impact of sensing and aggregation errors on learning performance. Fu \emph{et al.}~\cite{Fu2025JSCC} studied a joint sensing--communication--computation (JSCC) design for aerial-network-enabled OTA-FL, where sensing, communication, and computation resources are optimized under aerial-network constraints. More recently, Wen \emph{et al.}~\cite{Wen2026ISCCAirFEEL} developed an ISCC-enabled OTA-FL framework that characterizes the effects of sensing noise and aggregation distortion.

Sensing-aware FL has also been studied in vertical or feature-partitioned learning architectures. Liu \emph{et al.}~\cite{Liu2022VFEELISAC} proposed a vertical federated edge learning (V-FEEL) framework with distributed integrated sensing and communication, where devices observe the same target from different sensing views. Cao \emph{et al.}~\cite{Cao2026VFEEL} further developed an ISCC-enabled V-FEEL framework for edge perception networks, where sensing noise, multi-view feature diversity, and OTA embedding aggregation jointly affect learning performance. These V-FEEL works are closely related in motivation, but they focus on feature-partitioned data and embedding aggregation, whereas this paper considers horizontal OTA-FL with sample-partitioned local datasets and model-update aggregation. 

Despite these advances, existing sensing-aware FL formulations still represent sensing quality through a limited set of controls, such as sample count or sensing power. These variables are important, but they do not fully capture the sensing-quality tradeoffs that arise in diverse edge systems. In particular, two devices collecting the same number of samples may produce data of different fidelity due to heterogeneous sensors. Conversely, increasing sensing resolution or sensing power may reduce measurement distortion without increasing the sample count, at the cost of higher energy and latency.

\subsection{Contributions}
This paper addresses the above gap by adopting a sensing-centric view of OTA-FL, where sensing quality is treated as a multidimensional controllable component rather than an implicit property of the local dataset. Specifically, each device jointly selects its sensing modality, sensing resolution, sensing power, and per-round sample size, while local computation and OTA communication resources are optimized within the same learning-aware framework. This distinction is important because sensing distortion corrupts the training data before local optimization and therefore propagates through subsequent gradient updates, unlike communication noise that affects only the aggregation stage.

To capture this effect, we employ an additive sensing noise model that relates sensing configuration to data fidelity in a tractable manner. We then derive a convergence characterization that captures the combined effects of sensing-induced distortion, stochastic gradient variability, and OTA aggregation error. Guided by this analysis, we formulate a per-round mixed discrete--continuous optimization problem that balances learning quality, sensing energy, computation latency, and communication distortion. The problem is solved using a low-complexity alternating method with closed-form and one-dimensional updates for continuous variables and finite search over discrete sensing configurations.

The main contributions of this work are summarized as follows:
\begin{itemize}
    \item We introduce a sensing-quality-aware OTA-FL model that captures the impact of sensing modality, resolution, power, and sample size on data fidelity and learning performance.
    \item We derive a convergence characterization that jointly accounts for sensing-induced distortion, stochastic gradient variability, and OTA aggregation error.
    \item We formulate and solve a unified sensing--computation--communication optimization problem using a low-complexity alternating approach.
\end{itemize}

\emph{Notation:} matrices and vectors are denoted by bold capital and bold small letters, respectively. For a matrix $\bA$, the transpose is denoted $\bA^\top$, the Hermitian $\bA^H$, and the conjugate $\bA^*$. $\bI_m$ is the $m \times m$ identity matrix, and the dimension is omitted if clear from the context. Expectation is denoted $\mathbb{E}(\cdot)$.

\section{System Model}
\label{sec:system-model}
In this section, we describe the system model for a general distributed learning paradigm over wireless channels, depicted in Fig.~\ref{fig:scenarion}. We present the models for \emph{global model broadcast, local data sensing, local computation, uplink channel estimation, and uplink model aggregation}.

\begin{figure}[t]
    \centering
    \includegraphics[page=2,width=\Figwidth]{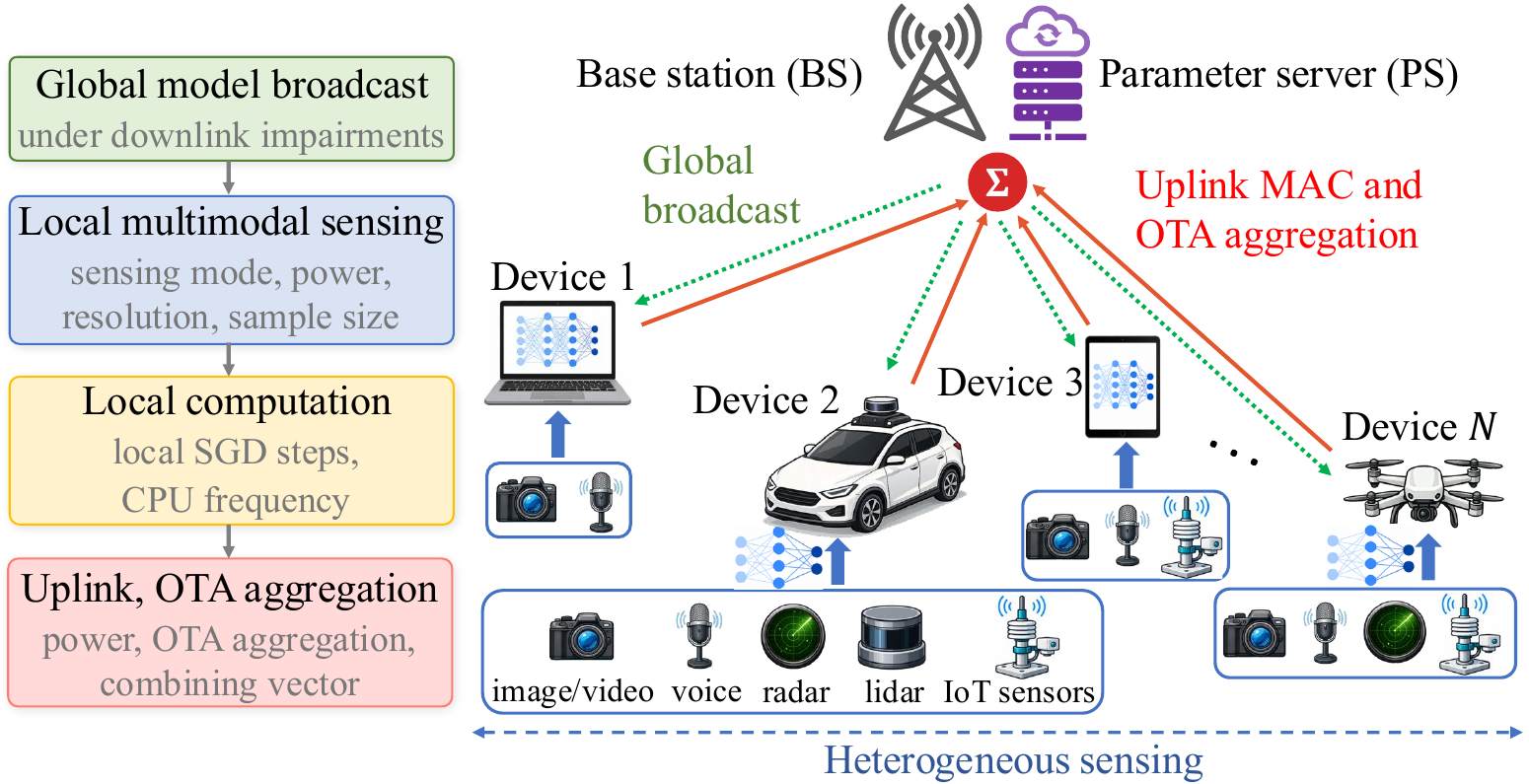}
    \caption{System overview of the joint sensing, computation, and communication design for OTA-FL under heterogeneous sensing and wireless link impairments.}
    \label{fig:scenarion}
    \vspace{-1.3em}
\end{figure}

\subsection{Preliminaries}
We consider a federated learning system with a parameter server (PS) and $N$ participating devices (see Fig.~\ref{fig:scenarion}). The server is located at the base station (BS), which has $M$ antennas. All devices have a single antenna. The model training process is carried out over $T$ communication rounds. Throughout the paper, we use the terms ``device'', ``receiver'', and ``user'' interchangeably to refer to the receiving terminals in the downlink. The system operates under physical, network, data, and computation resource constraints and heterogeneity.

Each device $n \in [N]$ begins with an initial local dataset $\mathcal{D}_{n,0}$ of size $D_{n,0} = |\mathcal{D}_{n,0}| \ge 0$, which may be zero. In each communication round~$t \in [T]$, device~$n$ senses a fresh batch of $\overline{D}_{n,t} = |\overline{\mathcal{D}}_{n,t}|$ new samples through the local sensing process described in Section~\ref{subsec:sensing}. These newly acquired samples are appended to the device's accumulated dataset:
\begin{align}
\begin{cases}
\mathcal{D}_{n,t} = \mathcal{D}_{n,t-1} \cup \overline{\mathcal{D}}_{n,t} \\
D_{n,t} = D_{n,t-1} + \overline{D}_{n,t}
\end{cases}.
\label{eq:data-accumulation}
\end{align}

Define $D_t \triangleq \sum_{n=1}^{N} D_{n,t}$, and let $f$ denote the per-sample loss function. The local empirical loss at device~$n$ in round~$t$ for a $d$-dimensional model $\btheta \in \mathbb{R}^d$ is
\begin{align}
\widetilde{F}_{n,t}(\btheta) \triangleq \frac{1}{D_{n,t}} \sum_{\boldsymbol{u} \in \mathcal{D}_{n,t}} f(\btheta, \boldsymbol{u}),
\label{eq:local-loss}
\end{align}
and the corresponding global loss is
\begin{align}
\widetilde{F}_t(\btheta) = \sum_{n=1}^{N} \frac{D_{n,t}}{D_t} \, \widetilde{F}_{n,t}(\btheta).
\label{eq:global-loss}
\end{align}
Note that both $\widetilde{F}_{n,t}$ and $\widetilde{F}_t$ evolve deterministically with the accumulation history $\{\overline{\mathcal{D}}_{n,\ell}\}_{\ell \le t}$.

FL aims to minimize $\widetilde{F}_t(\btheta)$ through iterative collaboration between the PS and the $N$ edge devices. 
During each round, the BS broadcasts the global model $\btheta_t$ to the devices over the downlink channel, which is subject to link impairments. Each device~$n$ receives a noisy version of the model, denoted by $\btheta_{n,t}$, and performs $\tau$ steps of mini-batch stochastic gradient descent (SGD) using its accumulated local data. At local step $i \in [\tau]$, device~$n$ draws a mini-batch $\bxi_{n,t}^{(i)}$ of fixed size $B$ uniformly from $\mathcal{D}_{n,t}$ and updates its local model as
\begin{align}
\btheta_{n,t}^{(i+1)} = \btheta_{n,t}^{(i)} - \eta \, \nabla \widetilde{F}_{n,t}\!\Big(\btheta_{n,t}^{(i)}; \bxi_{n,t}^{(i)}\Big), \quad i \in [\tau],
\label{eq:local-updates}
\end{align}
where $\btheta_{n,t}^{(1)} = \btheta_{n,t}$ and $\eta > 0$ is the global stepsize. After $\tau$ local iterations, device~$n$ forms its model update
\begin{align}
\Delta\btheta_{n,t} \triangleq \btheta_{n,t}^{(\tau+1)} - \btheta_{n,t}^{(1)}.
\label{eq:local-update-def}
\end{align}
Each device $n$ simultaneously transmits a scaled version of its $\Delta\btheta_{n,t}$ to the BS over the uplink, as described in Sections~\ref{subsec:pilot} and~\ref{subsec:ota}. The PS employs the received superposition to update the global model, which in effect becomes
\begin{align}
\btheta_{t+1} = \btheta_t + \sum_{n=1}^{N} \frac{D_{n,t}}{D_t} \, \widehat{\Delta\btheta}_{n,t},
\label{eq:model-update}
\end{align}
where $\widehat{\Delta\btheta}_{n,t}$ is a noisy version of each device's update.
This process continues over $T$ training rounds.

\subsection{Global Model Broadcast}
\label{subsec:downlink}
In each communication round~$t \in [T]$, the BS sends the current global model $\btheta_t$ to all devices over the downlink. Let $\bh_{n,t} \in \mathbb{C}^M$ denote the channel between the BS and device~$n$ in round~$t$, with entries $\mathcal{CN}(0,1)$. Further, let $S_m$ and $W^{\rm dl}$ denote the number of bits\footnote{The value of $S_m$ depends on the transmission mode (analog or digital), coding rate, quantization level, and other practical factors that determine how it relates to the model dimension $d$. A detailed discussion of these aspects is outside the scope of this work; see~\cite{Amiri2020FLFading} for more information.} and the downlink bandwidth, respectively. 

The received signal at device~$n$ in round~$t$ is
\begin{align}
y_{n,t} = 
\bh^{H}_{n,t} \bx^{\rm dl}_t + w_{n,t},
\end{align}
where $\bx^{\rm dl}_t \in \mathbb{C}^M$ is the downlink transmit signal across $M$ antennas, satisfying the power constraint $\mathbb{E}[(\bx^{\rm dl}_t)^H \bx^{\rm dl}_t] = p^{\rm dl}$ at the BS, and $w_{n,t} \sim \mathcal{CN}(0,\sigma_w^2)$ is the AWGN at device~$n$. The corresponding per-device downlink SNR is
\begin{align}
\mathrm{SNR}^{\rm dl}_{n,t}
=
\frac{\|\bh_{n,t}\|^2 p^{\rm dl}}{\sigma_w^2}.
\end{align}
Since the BS uses a common broadcast transmission, the achievable downlink rate in round~$t$ is limited by the weakest link (device). Therefore,
\begin{align}
R^{\rm dl}_{t}
&=
W^{\rm dl}\,
\log_2\!\Big( 1 + \min_{n \in [N]} \mathrm{SNR}^{\rm dl}_{n,t} \Big).
\end{align}
The time required to broadcast the global model to all devices in round~$t$ is
\begin{align}
T^{\rm dl}_{t}
&=
\frac{S_m}{R^{\rm dl}_{t}}
=
\frac{S_m}
{W^{\rm dl}\,
\log_2\!\Big( 1 + \min_{n} \mathrm{SNR}^{\rm dl}_{n,t} \Big)}.
\label{eq:downlink-time}
\end{align}
The corresponding downlink energy at the BS is
\begin{align}
E^{\rm dl}_{t}= p^{\rm dl} \, T^{\rm dl}_{t}.
\end{align}


\begin{remark}
\label{remark:TDD}
The system operates under time-division duplexing (TDD) mode with perfect reciprocity for simplicity. Therefore, we do not consider downlink channel training, as the BS can estimate the channels via uplink pilots. Our framework is also applicable to frequency-division duplexing (FDD) systems, with additional overhead for downlink training and feedback.
\end{remark}

\subsection{Local Data Sensing}
\label{subsec:sensing}
In the downlink round~$t$, after receiving $\btheta_t$, each device~$n$ performs local sensing to collect new data for its local training. We consider a finite set of sensing modalities
\begin{align*}
\mathcal{M} = \{1, \cdots, M_{\mathrm{mod}}\},
\end{align*}
which may represent, e.g., camera, microphone, LiDAR, and RF sensing. Device~$n$ selects exactly one modality in round~$t$, from $\mathcal{M}_n \subseteq \mathcal{M}$, using binary selection variables $z^{(m)}_{n,t} \in \{0,1\}$ such that
\begin{align*}
\sum_{m \in \mathcal{M}_n} z^{(m)}_{n,t} = 1,
\quad \forall n,t.
\end{align*}

Each device~$n$ also chooses a sensing resolution $q_{n,t}$ (the number of measurements captured per sample, e.g., pixels per frame or points per scan) from a device-specific finite set
\begin{align*}
\mathcal{Q}_n = \{q^{(1)}_n,\ldots,q^{(Q_n)}_n\}.
\end{align*}

Recall that $\overline{D}_{n,t}$ denotes the size of newly sensed data in round~$t$ at device~$n$. Let $\bs_{n,t}$ denote the underlying physical signal. The sensed sample (feature vector) is modeled as
\begin{align}
\bx_{n,t} 
&= \bs_{n,t} + \bv_{n,t},
\label{eq:sensing-additiveNoise}
\end{align}
where $\bv_{n,t}$ is the additive sensing noise distributed~$\bv_{n,t} \sim \mathcal{N}\big(0, {\sigma}^{2}_{v,n,t} \mathbf{I}\big)$. The noise variance depends on the selected modality, resolution, and sensing power. Let $p^{ s}_{n,t} \in (0, P_{n}^{s,\max}]$ denote the allocated power for sensing at device~$n$ in round~$t$, where $P_{n}^{s,\max}$ is the maximum sensing power supported by the sensor front-end. We model the sensing noise variance as
\begin{align}
{\sigma}^{2}_{v,n,t}=
\sum_{m \in \mathcal{M}_n} 
z^{(m)}_{n,t}
\frac{c^{(m)}_{1,n}}{p^{ s}_{n,t} \,q_{n,t}},
\label{eq:sensing-noise}
\end{align}
where $c^{(m)}_{1,n} > 0$ is a modality- and device-specific constant characterizing the sensing hardware (e.g., sensor sensitivity and front-end noise figure). This captures the natural trade-off that higher sensing power and higher resolution yield lower sensing noise.

The time to sense $\overline{D}_{n,t}$ samples under modality~$m$ depends on a modality-specific sensing rate $r^{(m)}_{{ s},n} > 0$ (samples per second) and on the acquisition resolution $q_{n,t}$: finer resolution captures more measurements per sample, lowering the effective sampling rate to $r^{(m)}_{{ s},n}/q_{n,t}$. The sensing latency is
\begin{align}
T^{ s}_{n,t}
&= \sum_{m \in \mathcal{M}_n} 
z^{(m)}_{n,t}
\frac{\overline{D}_{n,t}\,q_{n,t}}{r^{(m)}_{{ s},n}},
\label{eq:sensing-time}
\end{align}
and the corresponding sensing energy is
\begin{align}
E^{ s}_{n,t}
&= p^{ s}_{n,t} T^{ s}_{n,t}.
\label{eq:sensing-energy}
\end{align}
After sensing, the newly acquired samples $\overline{\mathcal{D}}_{n,t}$ are added to the existing dataset $\mathcal{D}_{n,t-1}$ to form $\mathcal{D}_{n,t}$.

\begin{remark}
The additive Gaussian sensing noise model in~\eqref{eq:sensing-additiveNoise} and \eqref{eq:sensing-noise} is adopted as a tractable abstraction for capturing the dependence of sensing fidelity on modality, power, and resolution. While certain modalities exhibit structured or signal-dependent noise (e.g., Poisson statistics in optical sensing or multiplicative noise in radar), the Gaussian model is the entropy-maximizing choice among additive noises with the same variance~\cite{Diggavi2001} and is adopted here as a tractable information-theoretic worst case. Extending the analysis to signal-dependent noise models is left to future work.
\end{remark}

\subsection{Local Computation Model}
Following the local SGD updates in~\eqref{eq:local-updates} and the model update definition in~\eqref{eq:local-update-def}, we now model the latency and energy of the local computation phase at each device. Let $f^{ c}_{n,t}$ denote the CPU frequency (cycles per second) at device~$n$ in round~$t$, and let $\kappa$ denote the number of CPU cycles required to compute one sample gradient (a system constant determined by the model architecture and SGD configuration). Since device~$n$ performs $\tau$ local SGD steps in round~$t$, each on a mini-batch of size $B$, the total number of sample-gradient evaluations per round is $\tau B$. The local computation time is therefore
\begin{align}
T^{c}_{n,t} = \frac{\kappa \tau B}{f^{c}_{n,t}}.
\label{eq:comp-time}
\end{align}
We adopt the standard dynamic voltage and frequency scaling (DVFS) model for processor energy~\cite{Tran2019FL,Yang2020EE}, under which the energy consumed for a workload of $\kappa \tau B$ cycles at frequency $f^{ c}_{n,t}$ is
\begin{align}
E^{ c}_{n,t} = \kappa_c \, \kappa \, \tau \, B \, \big(f^{ c}_{n,t}\big)^2,
\label{eq:comp-energy}
\end{align}
where $\kappa_c > 0$ is the effective switched capacitance coefficient of the processor. The CPU frequency at each device is constrained as
\begin{align}
f^{\min}_n \le f^{ c}_{n,t} \le f^{\max}_n, \quad \forall n,t.
\label{eq:cpu-freq-constraint}
\end{align}
The computation cost depends on $B$, $\tau$, and $f^{c}_{n,t}$, but not on the accumulated dataset size $D_{n,t}$. The newly sensed sample count $\overline{D}_{n,t}$ remains a sensing-side decision variable that affects the sensing latency and energy in~\eqref{eq:sensing-time}--\eqref{eq:sensing-energy} and the accumulated size $D_{n,t}$ entering the convergence analysis in Section~\ref{sec:convergence}.

\subsection{Uplink Channel Training}
\label{subsec:pilot}
Before transmitting the local model updates, the BS acquires channel state information (CSI) via uplink pilot transmission. 
At the beginning of round~$t$, each device transmits an uplink pilot sequence of length $\pi_p$ symbols. We assume orthogonal pilots (no pilot contamination). Let $p^{\rm p}_{n,t}$ denote the pilot transmit power of device~$n$ in round~$t$. The received pilot signal corresponding to device~$n$ at the BS is
\begin{align}
\by^{\rm p}_{n,t} =
\sqrt{\pi_p p^{\rm p}_{n,t}}\, \bh_{n,t} + \bz_t,
\label{eq:pilot-observation}
\end{align}
where $\bz_t \sim \mathcal{CN}(0, \sigma_z^2 \mathbf{I}_M)$ is the AWGN in the received signal at the BS. Observing~\eqref{eq:pilot-observation}, the MMSE estimate of $\bh_{n,t}$ is \cite{Hassibi2003HowMuch}
\begin{align}
\widehat{\bh}_{n,t} = \frac{\sqrt{\pi_p p^{\mathrm{p}}_{n,t}}}{\pi_p p^{\mathrm{p}}_{n,t} + \sigma_z^2} \, \by^{\mathrm{p}}_{n,t}.
\label{eq:mmse-estimate}
\end{align}
The corresponding estimation error is
\begin{align}
\be_{n,t}
&\triangleq 
\bh_{n,t} - \widehat{\bh}_{n,t},
\label{eq:MMSE-error}
\end{align}
which is Gaussian and uncorrelated with the channel estimate $\widehat{\bh}_{n,t}$~\cite{Hassibi2003HowMuch}. Therefore,
\begin{align}
\mathbb{E}\big[\widehat{\bh}_{n,t}\widehat{\bh}_{n,t}^H\big] &= (1- \sigma_{e,n,t}^2)\mathbf{I}_M, \\
\mathbb{E}\big[\be_{n,t}\be_{n,t}^H\big] &= \sigma_{e,n,t}^2 \mathbf{I}_M,
\end{align}
where the estimation error variance is
\begin{align}
\sigma_{e,n,t}^2= 1 - \frac{\pi_p p^{\mathrm{p}}_{n,t} } {\pi_p p^{\mathrm{p}}_{n,t} + \sigma_z^2} = \frac{\sigma_z^2}{\pi_p p^{\mathrm{p}}_{n,t} + \sigma_z^2}.
\label{eq:est-error-variance}
\end{align}


\subsection{Uplink OTA Aggregation of Model Updates}
\label{subsec:ota}
Each device~$n$ transmits its model update $\Delta\btheta_{n,t}$ to the BS in the uplink. We focus on a per-coordinate description of the transmission, which is standard in OTA-FL: each element of $\Delta\btheta_{n,t}$ is transmitted in a symbol interval, and the OTA aggregation is performed coordinate-wise. For brevity, we consider one generic coordinate, denoted by $\Delta\theta_{n,t}$, and omit the coordinate index. We assume that the update coordinates are power-normalized, with the normalization/de-normalization factors absorbed into the transmit/receive scaling and omitted from the notation.

Let $p^{u}_{n,t} \in [0, P_{n}^{u,\max}]$ denote the transmit power used by device~$n$ in round~$t$, where $P_{n}^{u,\max}$ is the maximum transmit power of the device's power amplifier. The received signal at the BS across its $M$ antennas is
\begin{align}
\by_t &= \sum_{n=1}^{N} \bh_{n,t} \sqrt{p^{u}_{n,t}}\, \Delta\theta_{n,t} + \bz_t,
\label{eq:uplink-signal}
\end{align}
where $\bz_t \sim \mathcal{CN}(0, \sigma_z^2 \mathbf{I}_M)$. Using \eqref{eq:MMSE-error}, \eqref{eq:uplink-signal} becomes
\begin{align}
\by_t = \sum_{n=1}^{N} \widehat{\bh}_{n,t} \sqrt{p^{u}_{n,t}}\, \Delta\theta_{n,t} + \sum_{n=1}^{N} \be_{n,t} \sqrt{p^{u}_{n,t}}\, \Delta\theta_{n,t} + \bz_t.
\end{align}
The BS applies a linear combining vector $\bv_t \in \mathbb{C}^{M \times 1}$ to obtain a scalar aggregated estimate
\begin{align}
\tilde{y}_t = \bv_t^H \by_t.
\end{align}

The desired aggregated update for the considered coordinate is the weighted sum
\begin{align}
g_t \triangleq 
\sum_{n=1}^{N} \alpha_{n,t} \Delta\theta_{n,t},
\label{eq:g_t}
\end{align}
where $\alpha_{n,t} \triangleq D_{n,t}/D_t$. Due to channel distortion, estimation errors, and noise, $\tilde{y}_t$ will deviate from $g_t$. We define the mean-square OTA aggregation error (per coordinate) as
\begin{align}
\mathrm{MSE}^{u}_t \triangleq 
\mathbb{E}\big[\,|\tilde{y}_t - g_t|^2 \,\big|\, \{\widehat{\bh}_{n,t}\}_{n=1}^{N} \big],
\end{align}
where the expectation is taken over the channel estimation error, receiver noise, and normalized update coordinates. Specifically, we assume that the normalized update coordinates have unit second moment and are second-order orthogonal across devices:
\begin{align}
\mathbb{E}[|\Delta\theta_{n,t}|^2]=1,\quad \mathbb{E}[\Delta\theta_{n,t}\Delta\theta_{j,t}^{*}]=0, \,\, n\neq j,\,\, \forall n,t.
\label{eq:variance-of-update-elements}
\end{align}
The normalized update coordinates are also independent of the channel estimation error and receiver noise. Under these assumptions and the MMSE channel estimation model in~\eqref{eq:mmse-estimate}-\eqref{eq:est-error-variance}, we have 
\begin{align}
\mathrm{MSE}^{u}_t = &
\sum_{n=1}^{N}
\big| \bv_t^H \widehat{\bh}_{n,t} \sqrt{p^{u}_{n,t}} - \alpha_{n,t} \big|^2 \nonumber \\
&+ \|\bv_t\|^2 \Big( \sum_{n=1}^{N} p^{u}_{n,t}\sigma_{e,n,t}^2 + \sigma_z^2 \Big).
\label{eq:mse-ota-simplified}
\end{align}
In the sequel, we will use \eqref{eq:mse-ota-simplified} as the OTA aggregation distortion metric that enters the convergence bound and optimization problem.

In analog OTA aggregation, each of the $d$ model coordinates occupies one uplink symbol interval, so the uplink transmission time is common across devices and given by
\begin{align}
T^{u}_{t} &= \frac{d}{W^{u}}, \\
E^{u}_{n,t} &= p^{u}_{n,t} T^{u}_{t},
\end{align}
where $W^{u}$ is the uplink symbol rate.

\subsection{Per-Round Latency and Energy}
In round $t$, the total latency includes the downlink broadcast time, the maximum sensing and computation time across devices, and the common uplink aggregation time. Since the pilot length and pilot powers are fixed and are not optimization variables, the corresponding channel estimation latency and energy are constant with respect to the design variables and are omitted from the per-round accounting. Hence, the per-round latency is
\begin{align}
T_t = T_t^{\rm dl} + \max_{n}\big(T^{s}_{n,t}+T^{c}_{n,t}\big)+T^{u}_{t}.
\label{eq:round-latency}
\end{align}
The total energy consumption in round~$t$ is the sum of the devices' sensing, computation, and uplink communication energies. Since the downlink energy is typically dominated by the BS and does not impact device-side optimization, we focus on the device energies and define
\begin{align}
E_t = \sum_{n=1}^{N} \Big(E^{s}_{n,t}+E^{c}_{n,t}+E^{u}_{n,t}\Big).
\label{eq:round-energy}
\end{align}
We will later leverage the expressions in~\eqref{eq:round-latency} and \eqref{eq:round-energy}, together with a convergence bound on the FL process, to formulate and solve a joint sensing--computation--communication optimization problem.


\section{Convergence Analysis}
\label{sec:convergence}
In this section, we study the learning performance of the proposed framework when the global loss is non-convex. For each sample vector $\bu \in \mathcal{D}_{n,t}$, let $\bs(\bu)$ denote its underlying clean signal in the sensing model~\eqref{eq:sensing-additiveNoise}. We define the clean reference local and global losses as
\begin{align} F_{n,t}(\btheta) &\triangleq \frac{1}{D_{n,t}} \sum_{\bu \in \mathcal{D}_{n,t}} f(\btheta,\bs(\bu)),\\ F_t(\btheta) &\triangleq \sum_{n=1}^{N}\frac{D_{n,t}}{D_t}F_{n,t}(\btheta). \label{eq:clean-global-loss} 
\end{align}
These losses are used only for analysis, while the actual local updates use $\widetilde{F}_{n,t}$ in~\eqref{eq:local-loss}. Since $F_t$ evolves with the accumulated data, we track $\frac{1}{T}\sum_{t=0}^{T-1}\mathbb{E}\!\left[\|\nabla F_t(\btheta_t)\|^2\right]$.

\subsection{Preliminaries and Assumptions}
\label{subsec:assumptions}
We work under standard OTA-FL assumptions~\cite{Cao2022JSACPowerControl,Sery2021OTAHetero,Amiri_TWC2021_FLConvergence}, together with sensing-specific conditions which apply uniformly across rounds.

\begin{enumerate}[leftmargin=*,label=\textbf{A\arabic*.}]

\item \textbf{$L$-smoothness:} The per-sample loss $f(\btheta, \bu)$ is $L$-smooth in $\btheta$ for every $\bu$,
\begin{align}
\|\nabla f(\btheta, \bu) - \nabla f(\btheta', \bu)\| \le L \|\btheta - \btheta'\|, \quad \forall \btheta, \btheta', \bu.
\end{align}
Hence $F_{n,t}$ and $F_t$ are $L$-smooth in $\btheta$ for every $n$ and $t$.

\item \textbf{Bounded loss:} There exist $F_{\inf}, F_{\sup} \in \mathbb{R}$ such that $F_{\inf} \le f(\btheta, \bs) \le F_{\sup}$ for every clean sample $\bs$ and all $\btheta$; consequently
\begin{align}
F_{\inf} \le F_t(\btheta) \le F_{\sup}, \quad \forall \btheta, \, t.
\end{align}
Let $\Delta F \triangleq F_{\sup} - F_{\inf}$.

\item \textbf{Bounded gradient norm:} There exists $G > 0$ such that
\begin{align}
\|\nabla F_{n,t}(\btheta;\bxi)\|\le G,
\quad \forall \btheta,n,t,\bxi,
\end{align}
where $\bxi$ is any admissible clean mini-batch.


\item \textbf{Mini-batch variance decomposition:} Let $F_n^{\infty}(\btheta) \triangleq \mathbb{E}_{\bs \sim \mathcal{P}_n}[f(\btheta, \bs)]$ denote the clean population loss at device $n$. For a mini-batch $\bxi^{(i)}_{n,t}$, $i \in [\tau]$, with size $B$ drawn uniformly from $\mathcal{D}_{n,t}$, the stochastic gradient error around the population gradient satisfies
\begin{align}
\mathbb{E}\!\left[\|\nabla F_{n,t}(\btheta; \bxi^{(i)}_{n,t}) - \nabla F_n^{\infty}(\btheta)\|^2\right] \le \frac{\sigma_B^2}{B} + \frac{\sigma_{\rm pop}^2}{D_{n,t}},
\end{align}
where $\sigma_B^2$ bounds the per-sample variance around the empirical mean and $\sigma_{\rm pop}^2$ bounds the per-sample variance around the population mean.

\item \textbf{Lipschitz-in-input gradient:} For each device $n$, there exists $c_n^s > 0$ such that, for any clean signal $\bs$ and any sensing-noise realization $\bv$,
\begin{align}
\mathbb{E}_{\bv}\!\left[\|\nabla f(\btheta; \bs + \bv) - \nabla f(\btheta; \bs)\|^2\right] \le c_n^s\, \mathbb{E}\|\bv\|^2.
\end{align}
This holds when $f(\btheta,\cdot)$ has a uniformly bounded cross-Jacobian $\nabla_{\bs}\nabla_{\btheta} f$, consistent with Jacobian-based sensitivity analyses under input perturbations~\cite{rifai2011adding}.

\item \textbf{First-order unbiased sensing perturbation:} Under the additive sensing model in~\eqref{eq:sensing-additiveNoise}, the zero-mean sensing noise induces a sample-gradient perturbation that is unbiased to first order,
\begin{align}
\mathbb{E}_{\bv}\!\left[\nabla f(\btheta; \bs + \bv) - \nabla f(\btheta; \bs)\right] = \bzero,
\end{align}
consistent with first-order analyses of learning under additive input noise~\cite{Bishop-NoisyTraining,rifai2011adding}.
\end{enumerate}

\begin{remark}
\label{remark:perfect-dl-for-conv}
Throughout the convergence analysis, we assume $\btheta_{n,t} = \btheta_t$ for every $n$ and $t$, i.e., the BS broadcast is decoded reliably at every device. This isolates the effects of sensing, local computation, and uplink OTA aggregation on the learning trajectory. Extending the analysis to the imperfect downlink case in Section~\ref{subsec:downlink} follows the same template: an additional perturbation term $\btheta_{n,t} - \btheta_t$ propagates through the $\tau$ local SGD steps and contributes a separate term to the per-round bound. We omit this for brevity.
\end{remark}

Let $\bzeta_t$ denote the aggregate perturbation, which in the per-round update arises from three physical sources: sensing noise in the local data acquisition phase, stochasticity and finite-data effects in the local computation phase, and uplink OTA aggregation distortion due to fading, imperfect CSI, and receiver noise at the BS. In the following, we derive the decomposition explicitly and separate each contribution into its bias and variance components.

\subsection{Decomposition of the Per-Round Update Error}
\label{subsec:perturbation-decomposition}
We begin by writing the BS-side update in a form that exposes the deviation from an ideal full-gradient descent step on $F_t$. Under Remark~\ref{remark:perfect-dl-for-conv}, each device initializes its local SGD at $\btheta^{(1)}_{n,t} = \btheta_t$. From~\eqref{eq:local-update-def} and~\eqref{eq:model-update}, we have
\begin{align}
\btheta_{t+1} - \btheta_t = \sum_{n=1}^{N} \frac{D_{n,t}}{D_t} \widehat{\Delta\btheta}_{n,t}.
\label{eq:bs-update-rearranged}
\end{align}
We add and subtract the ideal $\tau$-step full-gradient reference on $F_t$, namely $-\eta\tau\nabla F_t(\btheta_t)$, to obtain
\begin{align}
\btheta_{t+1} - \btheta_t = -\eta\tau\nabla F_t(\btheta_t) + \bzeta_t,
\label{eq:bs-update-perturbed}
\end{align}
where the aggregate perturbation $\bzeta_t$ is defined as
\begin{align}
\bzeta_t \triangleq \sum_{n=1}^{N} \frac{D_{n,t}}{D_t} \widehat{\Delta\btheta}_{n,t} + \eta\tau\nabla F_t(\btheta_t).
\label{eq:zeta-def}
\end{align}
The expression in~\eqref{eq:zeta-def} captures every deviation between the actual BS update and the ideal full-gradient step. We decompose $\bzeta_t$ along the three physical sources previewed in Section~\ref{subsec:assumptions} by inserting the noise-free local update as an intermediate quantity.

Let $\Delta\btheta_{n,t}^{\rm clean}$ denote the hypothetical local update that device~$n$ would have produced in the absence of sensing noise, using the same initialization and the same sequence of mini-batch indices as the actual local SGD, but replacing each noisy observation $\bx_{n,t}$ by its corresponding clean signal $\bs_{n,t}$. Adding and subtracting $\sum_{n} \frac{D_{n,t}}{D_t}\Delta\btheta_{n,t}$ and $\sum_{n} \frac{D_{n,t}}{D_t}\Delta\btheta_{n,t}^{\rm clean}$ in~\eqref{eq:zeta-def} results in~\eqref{eq:zeta-decomposition}, shown at the top of the next page.

\begin{table*}[!t]
\begin{align}
\bzeta_t = \underbrace{\sum_{n=1}^{N}\frac{D_{n,t}}{D_t}
\big(\widehat{\Delta\btheta}_{n,t} - \Delta\btheta_{n,t}\big)}
_{\triangleq \, \bepsilon_t} + \underbrace{\sum_{n=1}^{N}\frac{D_{n,t}}{D_t}
\big(\Delta\btheta_{n,t} - \Delta\btheta_{n,t}^{\rm clean}\big)}
_{\triangleq \, \bgamma^s_t} + \underbrace{\sum_{n=1}^{N}\frac{D_{n,t}}{D_t}
\Delta\btheta_{n,t}^{\rm clean}
+ \eta\tau\nabla F_t(\btheta_t)}
_{\triangleq \, \bgamma^c_t}
\label{eq:zeta-decomposition}
\end{align}
\hrulefill
\vspace{-1em}
\end{table*}

The three terms respectively capture the uplink OTA aggregation distortion ($\bepsilon_t$), the propagation of sensing noise through the local SGD updates ($\bgamma^s_t$), and the deviation of the noise-free $\tau$-step local SGD from the ideal full-gradient reference ($\bgamma^c_t$). Each contribution is expressed directly in terms of system quantities.

We separate each contribution into a deterministic given $\btheta_t$ bias and a zero-mean variance component. Define
\begin{align}
\bar{\bzeta}_t \triangleq \mathbb{E}\!\left[\bzeta_t \,\big|\, \btheta_t\right], \qquad \bzeta_t^{0} \triangleq \bzeta_t - \bar{\bzeta}_t,
\end{align}
so that $\mathbb{E}[\bzeta_t^{0} \,|\, \btheta_t] = \bzero$ and
\begin{align}
\mathbb{E}\!\left[\|\bzeta_t\|^2 \,\big|\, \btheta_t\right] = \|\bar{\bzeta}_t\|^2 + \mathbb{E}\!\left[\|\bzeta_t^{0}\|^2 \,\big|\, \btheta_t\right].
\label{eq:bias-variance-split}
\end{align}
Applying this split inside~\eqref{eq:zeta-decomposition}, the three contributions split as follows. By Assumption A6, the per-sample sensing perturbation is zero-mean in a first-order sense, so $\bgamma^s_t$ contributes only to the variance. The computation contribution $\bgamma^c_t$ has a nonzero conditional mean from client drift over the $\tau$ local steps on heterogeneous $F_{n,t}$, so it contributes to both the bias and the variance. The communication contribution $\bepsilon_t$ has a conditional mean whose power is the channel misalignment term in~\eqref{eq:mse-ota-simplified} and a zero-mean residual from the MMSE channel estimation error and the receiver noise.


\subsection{Bounding the Bias and Variance}
\label{subsec:bounding-perturbation}
We now bound each contribution in~\eqref{eq:zeta-decomposition} using Assumptions A1--A6, separating bias and variance where appropriate.

\subsubsection{Sensing contribution}
For each accumulated sample in $\mathcal{D}_{n,t}$, the noise variance is the one that prevailed in the round $\ell\le t$ when the sample was acquired:
\begin{align}
\sigma^2_{v,n,\ell} = \frac{c^{(m_{n,\ell})}_{1,n}}{p^s_{n,\ell}\, q_{n,\ell}}.
\end{align}
Let $\bgamma^s_{n,t}\triangleq \Delta\btheta_{n,t}-\Delta\btheta_{n,t}^{\rm clean}$. Since the actual and clean local updates use the same mini-batch indices, comparing their SGD recursions and applying A1 and A5 with $\eta\tau L\le 1/2$ gives
\begin{align}
\mathbb{E}\big\| \bgamma^s_{n,t} \big\|^2 \le 4\eta^2\tau^2\frac{c^s_n}{D_{n,t}} \sum_{\ell=1}^{t} \overline{D}_{n,\ell}\, \sigma^2_{v,n,\ell}.
\label{eq:sensing-perdev}
\end{align}
The factor $4$ follows by unrolling the noisy--clean SGD recursion and using $\sum_{i=0}^{\tau-1}(1+\eta L)^i \le \tau e^{\eta L\tau}<2\tau$ under $\eta\tau L\le 1/2$. Aggregating across devices with weights $\alpha_{n,t}=D_{n,t}/D_t$ and applying Jensen's inequality,
\begin{align}
\mathbb{E}\| \bgamma^s_t \|^2 \le 4\eta^2\tau^2 \sum_{n=1}^{N} \frac{D_{n,t}}{D_t}\, c^s_n\, \cdot \frac{1}{D_{n,t}} \sum_{\ell=1}^{t} \overline{D}_{n,\ell}\, \sigma^2_{v,n,\ell} \triangleq \mathcal{E}^s_t.
\label{eq:Ets-def}
\end{align}
The sensing floor depends on the entire history of sensing parameters used to acquire samples in $\mathcal{D}_{n,t}$, not just on the round-$t$ choices.

\subsubsection{Computation contribution}
The computation term $\bgamma^c_t$ measures the deviation of the noise-free $\tau$-step local SGD from the ideal full-gradient reference $-\eta\tau\nabla F_t(\btheta_t)$. We split it into a bias from client drift over heterogeneous $F_{n,t}$ and a zero-mean variance from mini-batch sampling.

\begin{lemma}[Local-step drift bound]
\label{lem:drift}
Under Assumptions A1 and A3 with the stepsize $\eta\tau L \le 1/2$, the conditional mean of $\bgamma^c_t$ satisfies
\begin{align}
\big\| \mathbb{E}[\bgamma^c_t \,|\, \btheta_t] \big\|^2 \le L^2\eta^4\tau^4 G^2.
\label{eq:drift-bound}
\end{align}
\end{lemma}

\begin{proof}
See Appendix~\ref{lemma:drift-proof}.
\end{proof}

\begin{lemma}[Mini-batch variance propagation]
\label{lem:variance}
Under Assumptions A1, A3, and A4 with the stepsize $\eta\tau L \le 1/2$, the variance of $\bgamma^c_t$ around its conditional mean satisfies
\begin{align}
\mathbb{E}\big\| \bgamma^c_t - \mathbb{E}[\bgamma^c_t \,|\, \btheta_t] \big\|^2 \le 3\eta^2 \tau  \sum_{n=1}^{N} \frac{D_{n,t}}{D_t} \left( \frac{\sigma^2_B}{B} + \frac{\sigma^2_{\rm pop}}{D_{n,t}} \right).
\label{eq:variance-bound}
\end{align}
\end{lemma}

\begin{proof}
See Appendix~\ref{lemma:variance-proof}.
\end{proof}

Combining Lemmas~\ref{lem:drift} and~\ref{lem:variance} via $\mathbb{E}\|\bgamma^c_t\|^2 = \|\mathbb{E}[\bgamma^c_t|\btheta_t]\|^2 + \mathbb{E}\|\bgamma^c_t - \mathbb{E}[\bgamma^c_t|\btheta_t]\|^2$, we have
\begin{align}
\mathbb{E}\| \bgamma^c_t \|^2 \le & \,\, L^2\eta^4\tau^4 G^2 \nonumber \\
&+ 3 \eta^2 \tau \sum_{n=1}^{N} \frac{D_{n,t}}{D_t}\!\left( \frac{\sigma^2_B}{B} + \frac{\sigma^2_{\rm pop}}{D_{n,t}} \right)  \triangleq \mathcal{E}^c_t.
\label{eq:Etc-def}
\end{align}
The bound exhibits three regimes: a local-step drift term controlled by $\eta\tau$, a mini-batch variance controlled by $B$, and an empirical-to-population variance that vanishes as $D_{n,t}$ grows.

\subsubsection{Uplink OTA contribution}
The OTA aggregation produces a noisy estimate of the weighted sum $g_t = \sum_n \alpha_{n,t}\Delta\theta_{n,t}$ through the per-coordinate model in Section~\ref{subsec:ota}. Under the normalized update model in~\eqref{eq:variance-of-update-elements}, the per-coordinate MSE in~\eqref{eq:mse-ota-simplified} aggregates over $d$ coordinates to
\begin{align}
\mathbb{E}\| \bepsilon_t \|^2 = d \cdot \mathrm{MSE}^{u}_t \triangleq \mathcal{E}^u_t.
\label{eq:Etu-def}
\end{align}
The factor $d$ converts the per-coordinate aggregation error in~\eqref{eq:mse-ota-simplified} into a model-dimension scaled gradient perturbation power. The first term in $\mathrm{MSE}^{u}_t$ is the channel misalignment error controlled by $(\bv_t, p^u_{n,t})$ jointly, and the second term is the variance from the MMSE channel estimation residual and receiver noise.

\subsubsection{Total perturbation power}

\begin{lemma}[Tightened total perturbation bound]
\label{lem:total-perturbation}
Under Assumptions A1--A6 and the stepsize $\eta\tau L\le 1/2$, with the sensing perturbation interpreted to first order as in A6, the second moment of the aggregate perturbation $\bzeta_t$ satisfies
\begin{align}
\mathbb{E}\| \bzeta_t \|^2 \le 2\mathcal{E}^s_t + 2\mathcal{E}^c_t + 2\mathcal{E}^u_t.
\label{eq:zeta-total-bound}
\end{align}
\end{lemma}

\begin{proof}
Let $\Xi_t$ denote the complete mini-batch index sequence in round $t$. Under the first-order sensing model of A6, conditioning on $\btheta_t$, the clean data, and $\Xi_t$ and linearizing the noisy--clean local recursion gives $\mathbb{E}[\bgamma_t^s\,|\,\btheta_t,\Xi_t]=0$. Since $\bgamma_t^c$ is fixed under this conditioning, the tower property yields $\mathbb{E}[\bgamma_t^{s\top}\bgamma_t^c\,|\,\btheta_t]=0$. Applying Young's inequality to $\bzeta_t=(\bgamma_t^s+\bgamma_t^c)+\bepsilon_t$ gives
\begin{align}
\|\bzeta_t\|^2\le 2\|\bgamma_t^s+\bgamma_t^c\|^2+2\|\bepsilon_t\|^2.
\end{align}
Taking expectations and using the above orthogonality,
\begin{align}
\mathbb{E}\|\bzeta_t\|^2 &\le 2\mathbb{E}\|\bgamma_t^s\|^2+2\mathbb{E}\|\bgamma_t^c\|^2+2\mathbb{E}\|\bepsilon_t\|^2 \notag\\
&\le 2\mathcal{E}^s_t+2\mathcal{E}^c_t+2\mathcal{E}^u_t.
\end{align}
This completes the proof.
\end{proof}

Define the per-round learning penalty as
\begin{align}
\Psi_t \triangleq 2\mathcal{E}^s_t + 2\mathcal{E}^c_t + 2\mathcal{E}^u_t,
\label{eq:Psi-def}
\end{align}
which gathers all the sensing, computation, and communication effects in round~$t$. The penalty is fully expressed in terms of system design variables: sensing history $\{(m_{n,\ell}, p^s_{n,\ell}, q_{n,\ell}, \overline{D}_{n,\ell})\}_{\ell\le t}$, local SGD configuration $(\tau, B, \eta)$, uplink powers $p^u_{n,t}$, MMSE channel estimates $\widehat\bh_{n,t}$ with error variances $\sigma^2_{e,n,t}$, and the BS combining vector $\bv_t$.

With the per-round penalty $\Psi_t$ in place, we now state the main convergence result.

\begin{theorem}
\label{thm:nonconvex}
Under Assumptions A1--A6 and the stepsize condition $\eta\tau L \le 1/2$, the iterates generated by the proposed scheme satisfy
\begin{align}
\frac{1}{T}\sum_{t=0}^{T-1} &\,\,\mathbb{E}\big[\|\nabla F_t(\btheta_t)\|^2\big] \le \frac{4\big(\mathbb{E}[F_0(\btheta_0)] - F_{\inf}\big)}{\eta\tau T} \nonumber \\
&+ \Big(\frac{4}{\eta^2\tau^2} + \frac{2L}{\eta\tau}\Big)\frac{1}{T}\sum_{t=0}^{T-1} \Psi_t + \frac{4}{\eta\tau T}\sum_{t=0}^{T-1}\delta_t,
\label{eq:nonconvex-main-bound}
\end{align}
where $\delta_t \triangleq \mathbb{E}\big[F_{t+1}(\btheta_{t+1}) - F_t(\btheta_{t+1})\big]$ is the per-round objective drift induced by the dataset accumulation from round $t$ to round $t+1$, bounded under A2 by $\delta_t \le \Delta F \sum_n \overline{D}_{n,t+1}/D_{t+1}$.
\end{theorem}

\begin{proof}
See Appendix~\ref{app:proof-nonconvex}.
\end{proof}

Inequality~\eqref{eq:nonconvex-main-bound} decomposes the average squared gradient norm into three parts. The first term is the optimization term, decaying as $\mathcal{O}(1/T)$ from the initialization gap. The second is the per-round learning penalty $\Psi_t$, which collects the sensing, computation, and communication distortions and does not vanish with $T$; it sets the accuracy floor that the system design controls. The third is the objective-drift term $\delta_t$, which arises because the global objective $F_t$ changes as devices accumulate freshly sensed data. Each term in the penalty maps to a distinct design lever: $\mathcal{E}^s_t$ decreases with improved sensing quality through higher sensing power and resolution; $\mathcal{E}^c_t$ shrinks with larger mini-batches, smaller local horizons, and larger datasets that reduce the empirical-to-population gap; and $\mathcal{E}^u_t$ shrinks with better channel estimation, uplink power control, and combining. The drift $\delta_t$ reflects the cost of changing the learning objective by sensing new data, and is bounded under A2 by the per-round change in the dataset composition.

\begin{remark}
\label{remark:drift-bounded}
The drift $\delta_t \le \Delta F \sum_n \overline{D}_{n,t+1}/D_{t+1}$ is set by the fraction of newly sensed data appended in round $t+1$. As the accumulated datasets grow, this fraction vanishes and the bound recovers the standard fixed-objective non-convex rate.
\end{remark}

\section{Problem Formulation and Solution Approach}
\label{sec:problem-solution}
In this section, we first formalize the optimization problem arising from the interaction of the sensing, computation, and communication modules described earlier. We then develop a structured solution method that exploits the problem’s separability and convexity properties.

Each communication round contributes to the global model update and is affected by three distortion sources: sensing noise, local computation limitations, and uplink OTA aggregation. These effects are captured through the per-round penalty $\Psi_t$, given in~\eqref{eq:Psi-def}, in which $\mathcal{E}_t^{s}$, $\mathcal{E}_t^{c}$, and $\mathcal{E}_t^{u}$ follow from the convergence bounds in Section~\ref{sec:convergence}.  
Improvements in sensing quality decrease $\mathcal{E}_t^{s}$, tighter SGD variance control decreases $\mathcal{E}_t^{c}$, and better OTA aggregation at the server reduces $\mathcal{E}_t^{u}$.



\subsection{Joint Optimization Problem}
Recall the total latency and energy requirements calculated in \eqref{eq:round-latency} and \eqref{eq:round-energy}, respectively. To balance energy consumption, latency, and learning quality, we consider a weighted per-round objective
\begin{align}
J_t = \omega_E E_t + \omega_T T_t + \omega_L \Psi_t,
\label{eq:Jt}
\end{align}
where the positive weights $(\omega_E,\omega_T,\omega_L)$ are set based on system priorities. The overall optimization problem then minimizes the sum of these per-round objectives over all communication rounds.

For each communication round~$t \in \{1, \ldots, T\}$, we collect all device- and BS-side optimization variables into
\begin{align}
\Theta_t
\triangleq
\big\{\overline{D}_{n,t},\, q_{n,t},\, z_{n,t}^{(m)},\, p_{n,t}^{s},\, f^{ c}_{n,t},\, p_{n,t}^{u} \big\}_{n=1}^{N} \cup \{\bv_t\}.
\label{eq:Thetat}
\end{align}
All of these variables are defined in Section~\ref{sec:system-model}. Uplink pilot power is not optimized and is assumed fixed during all channel estimation slots.

Let $\Theta \triangleq \{\Theta_t\}_{t=1}^{T}$. The joint optimization problem over $T$ rounds is
\begin{subequations}
\label{eq:main-opt}
\begin{align}
\underset{\Theta}{\min} \quad &
\sum_{t=1}^{T} J_t
\\
\text{s.t.} \quad &
0 < p_{n,t}^{s} \le P_{n}^{s,\max},
\\
& 0 \le p_{n,t}^{u} \le P_{n}^{u,\max},
\\
& f_{n}^{\min} \le f_{n,t}^c \le f_{n}^{\max},
\\
& \overline{D}_{n,t} \in \{1, \ldots, D_n^{\max}\},
\\
& q_{n,t} \in \mathcal{Q}_n,
\\
& z_{n,t}^{(m)} \in \{0,1\}, \quad \sum_{m \in \mathcal{M}_n} z_{n,t}^{(m)} = 1,
\\
& \bv_t \in \mathbb{C}^{M}.
\end{align}
\end{subequations}

The problem is non-convex due to the OTA aggregation term, the coupling between sensing parameters and learning performance, and the discrete nature of sample count, modality, and resolution choices. Since the rounds couple only through the accumulated dataset sizes determined by past decisions, we solve~\eqref{eq:main-opt} in a per-round manner, minimizing $J_t$ over $\Theta_t$ given the history from rounds $\ell<t$.

To develop a tractable solution, we group variables with similar functional roles as
\begin{align}
\Theta = \big\{ \Theta_{\mathrm{sens}},\, \Theta_{\mathrm{comp}},\, \Theta_{\mathrm{comm}},\, \Theta_{\mathrm{disc}} \big\},
\end{align}
where
\begin{align}
\begin{cases}
\Theta_{\mathrm{sens}} = \{p_{n,t}^{s}\}_{n,t} \\
\Theta_{\mathrm{comp}} = \{f^{ c}_{n,t}\}_{n,t} \\
\Theta_{\mathrm{comm}} = \{p_{n,t}^{u}\}_{n,t} \cup \{\bv_t\}_{t} \\
\Theta_{\mathrm{disc}} = \{\overline{D}_{n,t},\, q_{n,t},\, z_{n,t}^{(m)}\}_{n,t}
\end{cases}.
\label{eq:Theta-blocks}
\end{align}
Each block has a distinct functional role, and when the other blocks are fixed, the subproblem for each block becomes either convex or low-dimensional.

\subsection{Alternating Optimization Framework}
We adopt an alternating optimization (AO) method that updates one block at a time while holding the others fixed. This approach is effective because each block leads to a simplified subproblem that is either solvable in closed form or via a one-dimensional search.

\subsubsection{Sensing power update}
With $(\overline{D}_{n,t}, q_{n,t}, z_{n,t}^{(m)})$ fixed, the sensing contribution to $J_t$ reduces to a strictly convex function of $p_{n,t}^{s}$:
\begin{align}
\omega_E E_{n,t}^{s} + \omega_L \lambda_{n,t}^{s} \sigma_{v,n,t}^{2} = \omega_E p_{n,t}^{s} T_{n,t}^{s} + \omega_L \lambda_{n,t}^{s} \frac{c_{1,n}^{(m_{n,t})}}{p_{n,t}^{s} q_{n,t}},
\label{eq:ps}
\end{align}
where $\lambda_{n,t}^{s} \triangleq 8\eta^2\tau^2c_n^{s}\overline{D}_{n,t}/D_t$ is the per-device sensing coefficient in round~$t$ from~\eqref{eq:Ets-def}. Setting the derivative of \eqref{eq:ps} to zero yields the closed-form update
\begin{align}
p_{n,t}^{s\star} = \sqrt{
\frac{\omega_L \lambda_{n,t}^{s} c_{1,n}^{(m_{n,t})} }{\omega_E\, q_{n,t} T_{n,t}^{s}}},
\label{eq:ps-update}
\end{align}
followed by projection onto $(0,\,P_{n}^{s,\max}]$.

\subsubsection{ CPU frequency update}
The CPU frequencies are coupled through the maximum latency in~\eqref{eq:round-latency}. Introducing an auxiliary bottleneck latency $\chi_t$, the CPU frequency subproblem is
\begin{align}
\begin{aligned}
\underset{\{f^{\rm c}_{n,t}\}_{n=1}^{N},\,\chi_t}{\min}\quad &
\omega_E\kappa_c\kappa\tau B\sum_{n=1}^{N}\big(f^{ c}_{n,t}\big)^2+\omega_T\chi_t\\
\text{s.t.}\quad &
T^{s}_{n,t}+\frac{\kappa\tau B}{f^{c}_{n,t}}\le\chi_t, \\
&f_n^{\min}\le f^{c}_{n,t}\le f_n^{\max},\quad \forall n .
\end{aligned}
\label{eq:cpu-subproblem}
\end{align}
For a fixed feasible $\chi_t$, the energy-minimizing frequency is
\begin{align}
f_{n,t}^{c\star}(\chi_t)=\max\left\{f_n^{\min},\,\frac{\kappa\tau B}{\chi_t-T^{s}_{n,t}}\right\}.
\label{eq:f-update}
\end{align}
Hence, $\chi_t$ is obtained from the one-dimensional convex problem
\begin{align}
\chi_t^\star=
\underset{\chi_t^{\min} \leq\chi_t\leq\chi_t^{\max}}{\arg\min}\;
\omega_E\kappa_c\kappa\tau B
\sum_{n=1}^{N}\big(f_{n,t}^{ c\star}(\chi_t)\big)^2
+\omega_T\chi_t,
\label{eq:chi-update}
\end{align}
where $\chi_t^{\min}\triangleq\max_n\left(T^{s}_{n,t}+\kappa\tau B/f_n^{\max}\right)$ and $\chi_t^{\max}\triangleq\max_n\left(T^{s}_{n,t}+\kappa\tau B/f_n^{\min}\right)$. The scalar problem can be solved efficiently by one-dimensional search, after which $f_{n,t}^{ c\star}=f_{n,t}^{ c\star}(\chi_t^\star)$.


\subsubsection{Uplink power and combining vector update}
For fixed uplink powers, the OTA distortion is quadratic in $\bv_t$, yielding the linear MMSE-type update
\begin{align}
\bv_t = \left(\bA_t^{H}\bA_t + \gamma_t \bI \right)^{-1} \bA_t^{H}\balpha_t,
\label{eq:v-update}
\end{align}
where $\bA_t = [\sqrt{p_{1,t}^{u}}\widehat{\bh}_{1,t},\dots,\sqrt{p_{N,t}^{u}}\widehat{\bh}_{N,t}]^{H}$, $\balpha_t = [\alpha_{1,t},\dots,\alpha_{N,t}]^{H}$, and $\gamma_t = \sum_n p_{n,t}^{u}\sigma_{e,n,t}^{2} + \sigma_z^2$.

Given $\bv_t$, each device solves a one-dimensional optimization over $0 \le p_{n,t}^{u} \le P_{n}^{u,\max}$, which is a bounded subproblem that can
be handled efficiently by line search or projected Newton updates.

\subsubsection{Discrete variable updates}
We now describe how to update the discrete variables $(\overline{D}_{n,t}, q_{n,t}, z_{n,t}^{(m)})$ for each device. The sensed sample count is bounded as $1 \le \overline{D}_{n,t} \le D_n^{\max}$, and we write the candidate set as
\begin{align}
\mathcal{D}_n \triangleq \{1,2,\ldots,D_{n}^{\max}\},
\end{align}
so that the discrete design variables at device~$n$ are $m \in \mathcal{M}_n$, $q \in \mathcal{Q}_n$, and $\overline{D} \in \mathcal{D}_n$.

When updating device~$n$ in round~$t$, all other variables are held fixed. For a candidate count $\overline{D}$, the accumulated dataset size is $D_{n,t} = D_{n,t-1} + \overline{D}$ from~\eqref{eq:data-accumulation}, and the network total is $D_t = D_{n,t} + \sum_{j \neq n} D_{j,t}$. For a candidate triple $(m,q,\overline{D})$, the sensing latency and energy are
\begin{align}
T_{n,t}^{s}(m,q,\overline{D}) &= \frac{q\,\overline{D}}{r^{(m)}_{{\rm s},n}}, \\
E_{n,t}^{s}(m,q,\overline{D}) &= p_{n,t}^{s} T_{n,t}^{s}(m,q,\overline{D}),
\end{align}
and the sensing noise variance $\sigma_{v,n,t}^{2}(m,q)$ is given in~\eqref{eq:sensing-noise}.
The computation latency and energy follow the mini-batch model in~\eqref{eq:comp-time}--\eqref{eq:comp-energy} and are independent of $\overline{D}$:
\begin{align}
T_{n,t}^{c} = \frac{\kappa \tau B}{f^c_{n,t}}, \qquad
E_{n,t}^{c} = \kappa_c \kappa \tau B (f^c_{n,t})^{2}.
\end{align}
Define the bottleneck latency among all devices except device~$n$ as $\chi_{-n,t} \triangleq \max_{j\neq n}\left(T^s_{j,t}+T^c_{j,t}\right)$. The local surrogate discrete cost at device~$n$ in round~$t$ is
\begin{align}
J_{n,t}^{\mathrm{disc}}(m,q,\overline{D}) \triangleq &\,
\omega_E \big( E_{n,t}^{s}(m,q,\overline{D}) + E_{n,t}^{c} \big) \nonumber\\ 
&+\omega_T
\max\left\{
\chi_{-n,t},
T^s_{n,t}(m,q,\overline{D})+T^c_{n,t}
\right\} \nonumber\\
&+\omega_L \Big(
\lambda_{n,t}^{s}\, \sigma_{v,n,t}^{2}(m,q)+\frac{\lambda_{t}^{c}}{D_t}
\Big),
\end{align}
where $D_t = D_{n,t-1} + \overline{D} + \sum_{j \neq n} D_{j,t}$, and $\lambda_t^{c} \triangleq 6\eta^2\tau N\sigma_{\rm pop}^2$ from~\eqref{eq:Etc-def}. The first two terms account for the sensing and computation energy and the round bottleneck latency, while the last term reflects the impact of sensing noise and the accumulated sample count on the round-$t$ learning penalty $\Psi_t$.

The discrete update at device~$n$ is the finite optimization
\begin{align}
(m_{n,t}^{\star}, q_{n,t}^{\star}, \overline{D}_{n,t}^{\star}) = \underset{{\substack{m \in \mathcal{M}_n,\, q \in \mathcal{Q}_n,\, \overline{D} \in \mathcal{D}_n}}}{\arg\min}
J_{n,t}^{\mathrm{disc}}(m,q,\overline{D}).
\end{align}
This rule specifies how the modality, resolution, and number of newly sensed samples are chosen. Since $\mathcal{M}_n$, $\mathcal{Q}_n$, and $\mathcal{D}_n$ are finite and typically small, the optimization is carried out by direct search over all candidate triples $(m,q,\overline{D})$.

The overall AO procedure for the sensing, computation, and communication co-design is summarized in Algorithm~\ref{alg:ao}.

\begin{algorithm}[t]
\caption{Alternating Optimization for Joint Sensing--Computation--Communication Design}
\label{alg:ao}
\DontPrintSemicolon
\KwIn{System parameters, feasible sets $\{\mathcal{Q}_n,\mathcal{D}_n\}_{n=1}^N$, modality set $\mathcal{M}_n$, weights $(\omega_E,\omega_T,\omega_L)$, and constraints.}
\KwOut{Per-round decisions $\{\Theta_t\}_{t=1}^T$.}

Initialize $\{\Theta_t^{(0)}\}_{t=1}^T$\;
\For{$t=1,\ldots,T$}{
    $k\leftarrow 0$\;
    \Repeat{$J_t$ converges or $k=k_{\max}$}{
        Update $\{p^{s,(k+1)}_{n,t}\}_{n=1}^N$ by \eqref{eq:ps-update} and project onto $(0,P^{s,\max}_n]$\;
        Update $\chi_t^{(k+1)}$ by \eqref{eq:chi-update} and $\{f^{\rm c,(k+1)}_{n,t}\}_{n=1}^N$ by \eqref{eq:f-update} and project onto $[f_n^{\min},f_n^{\max}]$\;
        Update $\bv_t^{(k+1)}$ by \eqref{eq:v-update}\; 
        Update $\{p^{u,(k+1)}_{n,t}\}_{n=1}^N$ by solving the 1-D subproblems and projecting onto $[0,P^{u,\max}_n]$\;
        \For{$n=1,\ldots,N$}{
            $(m^\star,q^\star,\overline{D}^\star)
            \leftarrow \argmin 
            J^{\mathrm{disc}}_{n,t}(m,q,\overline{D})$\;
            Set $(z^{(m^\star)}_{n,t},q_{n,t},\overline{D}_{n,t})^{(k+1)}
            \leftarrow (1,q^\star,\overline{D}^\star)$, with other $z^{(m)}_{n,t}=0$\;
        }
        $k\leftarrow k+1$\;
    }
    $\Theta_t\leftarrow \Theta_t^{(k)}$\;
}

\end{algorithm}

\begin{figure*}[t]
    \centering

    \begin{subfigure}{0.3\textwidth}
        \centering
        \includegraphics[width=\columnwidth]{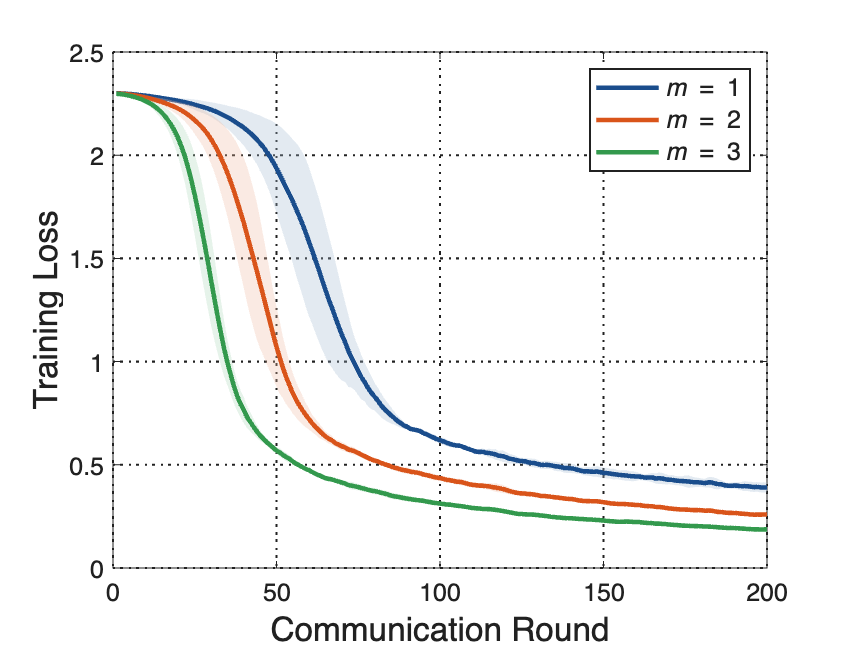}
    \caption{Modality training loss (MNIST).}
    \label{fig:train-loss}
    \end{subfigure}
    \hfill
    \begin{subfigure}{0.3\textwidth}
        \centering
        \includegraphics[width=\columnwidth]{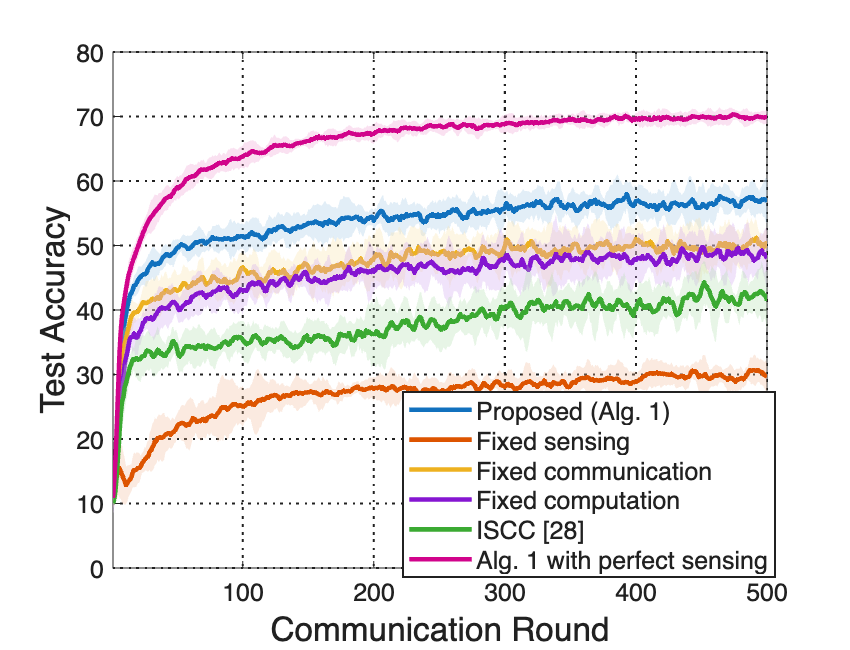}
        \caption{Test accuracy (CIFAR-10), i.i.d.}
        \label{fig:TestAcc-cifar10-iid}
    \end{subfigure}
    \hfill
    \begin{subfigure}{0.3\textwidth}
        \centering
        \includegraphics[width=\columnwidth]{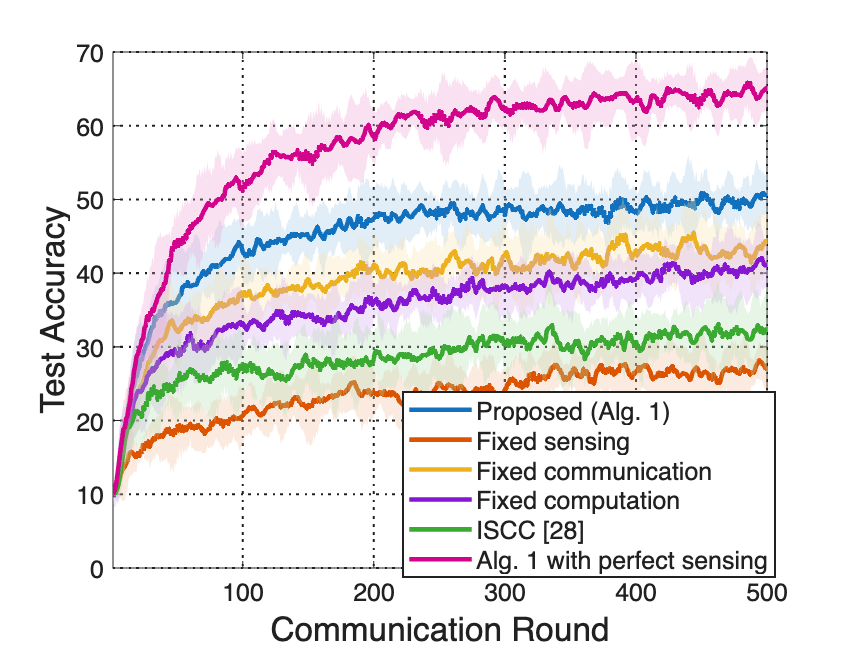}
        \caption{Test accuracy (CIFAR-10), non-i.i.d.}
        \label{fig:TestAcc-cifar10-noniid}
    \end{subfigure}

    \caption{Simulation results with the MNIST and CIFAR-10 datasets. (a) Training loss for the three sensing modalities in Table~\ref{tab:modality}. (b) Test accuracy for the proposed scheme and baselines with i.i.d. data partitioning. (c) Same as (b) for non-i.i.d. partitioning.}
    \label{fig:cifar10-results} 
      \vspace{-1em}
\end{figure*}

\subsection{Computational Complexity}


We analyze the per-round cost of Algorithm~\ref{alg:ao}. Let $K$ denote the number of AO iterations per round, and let $I_{c}$ and $I_{u}$ denote the numbers of iterations used by the one-dimensional searches for $\chi_t$ and $p^u_{n,t}$, respectively.

Within one AO iteration, the five blocks scale as follows. The sensing power update is closed-form and costs $\mathcal{O}(N)$, while the CPU frequency block requires a one-dimensional search over $\chi_t$ with $\mathcal{O}(N)$ cost per search iteration, yielding $\mathcal{O}(NI_{ c})$. The combining vector $\bv_t$ requires constructing an $M \times M$ Gram matrix from the $N$ estimated channels and solving one linear system, costing $\mathcal{O}(NM^2 + M^3)$. The uplink power update solves a bounded one-dimensional problem per device (convex in the amplitude $\sqrt{p^u_{n,t}}$), costing $\mathcal{O}(N I_{ u})$. The discrete update enumerates the triple $(m,q,\overline{D})$ for each device over $\mathcal{M}_n \times \mathcal{Q}_n \times \mathcal{D}_n$, costing $\mathcal{O}(N|\mathcal{M}||\mathcal{Q}||\mathcal{D}|)$, where $|\mathcal{M}| = \max_n |\mathcal{M}_n|$, $|\mathcal{Q}| = \max_n |\mathcal{Q}_n|$ and $|\mathcal{D}| = \max_n |\mathcal{D}_n|$. Over $K$ AO iterations, the per-round optimization cost is
\begin{align}
\mathcal{O}\!\Big( K \big(N I_{c}+NM^2+M^3+N I_{ u}+N|\mathcal{M}||\mathcal{Q}||\mathcal{D}|\big) \Big).
\end{align}
On top of this, each device runs $\tau$ local SGD steps on mini-batches of size $B$, costing $\mathcal{O}(N \tau B C_{\rm model})$ across the network in a round, where $C_{\rm model}$ is the cost of a single forward-backward pass of the model. At the BS, the OTA aggregation produces a length-$d$ estimate by applying $\bv_t$ to the received signal at cost $\mathcal{O}(Md)$ per round, and MMSE channel estimation from the pilot phase contributes $\mathcal{O}(NM\pi_p)$.



\section{Numerical Experiments}
\label{sec:numerical-results}

\subsection{Simulation Setup}
We conduct experiments using the MNIST~\cite{Lecun1998mnist}, CIFAR-10, and CIFAR-100~\cite{krizhevsky2009learning} datasets.
For training on MNIST, we use the convolutional neural network (CNN) architecture with fully connected layers. For CIFAR-10 and CIFAR-100, we respectively employ the ResNet-18 and ResNet-20 architectures. Unless stated otherwise, $M=8$, $N=20$, $\tau = 5$, $B=16$, $\sigma_w^2=\sigma_z^2 =1$, $D_n^{\max}=100$, $\mathcal{Q}_n = \{1,2,4\}$, $f_n^{\min}=0.5$ GHz, $f_n^{\max}=2$ GHz, and $P_n^{s,\max}=1$ W for all devices. The average SNR is set to $10$ dB for the downlink and uplink channels (both pilot and data phases). The downlink bandwidth is set to $10$ MHz, and the uplink symbol rate is set to $10$ Msymbol/s. Both i.i.d.\ and non-i.i.d.\ data distributions are considered across the devices. 


At each communication round~$t \in [T]$, device~$n$ selects a sensing configuration consisting of a modality, resolution, power, and the number of newly sensed samples. The sensing process follows the model in Section~\ref{subsec:sensing}: clean samples drawn from the underlying dataset are corrupted according to the selected sensing configuration and then used for local training. The sensing modalities are modeled as abstract configurations with different quality--latency tradeoffs, as summarized in Table~\ref{tab:modality}, where the rate $r_{\rm s}^{(m)}$ is measured in \emph{samples/s}. This captures heterogeneous sensing capabilities without tying the simulations to specific sensor physics.

Device sensing heterogeneity is incorporated by allowing each device to access only a subset of the available sensing modalities. Among the $N=20$ devices, $6$ have access to modality~$m=1$ only, $7$ can choose between $m=1$ and $m=2$, and the remaining $7$ have access to all three modalities.




\begin{table}
\centering
\caption{Sensing modality parameters used in simulations.}
\label{tab:modality}
\begin{tabular}{c c c l}
\hline
Modality $m$ & $c^{(m)}_{1}$ & $r_s^{(m)}$ & Description \\
\hline
1 & $1.0$ & $20$ & Low sensing quality, fast acquisition \\
2 & $0.5$ & $10$  & Moderate quality and acquisition speed \\
3 & $0.2$ & $1$  & High sensing quality, slow acquisition \\
\hline
\end{tabular}
\vspace{-1em}
\end{table}

\subsection{Results and Discussion}
Fig.~\ref{fig:cifar10-results}(a) shows the training loss for the three sensing modalities in Table~\ref{tab:modality} on the MNIST dataset under i.i.d.\ data distribution with $T=200$. Modality~$m=3$ achieves the fastest convergence and lowest steady-state loss, followed by $m=2$ and $m=1$. This confirms that higher sensing quality reduces data distortion, which improves local gradient estimation and speeds up convergence.

Fig.~\ref{fig:cifar10-results}(b) and Fig.~\ref{fig:cifar10-results}(c) show the test accuracy versus communication rounds ($T=500$) for the proposed scheme and benchmark schemes on the CIFAR-10 dataset under i.i.d.\ and non-i.i.d.\ data partitions. Several benchmark schemes are included to isolate the contribution of each resource-control block in the proposed joint design. Specifically, \emph{fixed communication} disables the communication-resource adaptation, \emph{fixed computation} disables the CPU frequency adaptation, \emph{fixed sensing} keeps the sensing configuration fixed (we assume $m=2$ for all devices) throughout training, and \emph{ISCC} optimizes the number of newly sensed samples along with the computation/communication parameters (the scheme given in~\cite{Liang2025ISCC}). The proposed joint sensing--computation--communication design achieves the highest accuracy among all practical schemes. For instance, in the non-i.i.d.\ case, the proposed scheme achieves approximately $18\%$ higher accuracy than the ISCC scheme in~\cite{Liang2025ISCC}. The proposed scheme with perfect sensing and perfect CSI serves as an oracle upper bound, corresponding to the ideal case in which the sensed data are free of measurement noise and the channel estimates are error-free.

\begin{figure*}[t]
    \centering

    \begin{subfigure}{0.3\textwidth}
        \centering
        \includegraphics[width=\columnwidth]{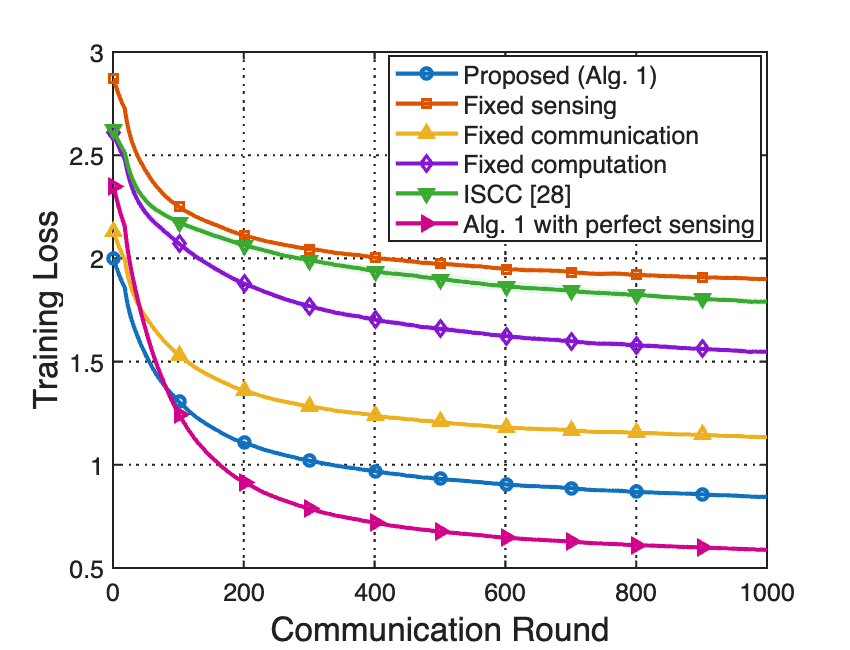}
    \caption{Training loss, i.i.d.}
    \label{fig:train-loss-cifar100}
    \end{subfigure}
    \hfill
    \begin{subfigure}{0.3\textwidth}
        \centering
        \includegraphics[width=\columnwidth]{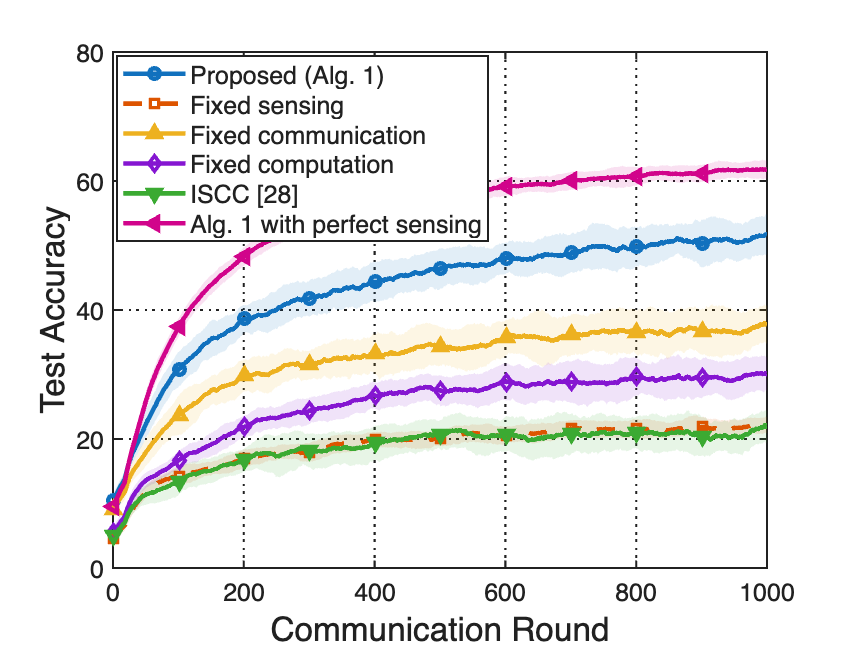}
        \caption{Test accuracy, non-i.i.d.}
        \label{fig:TestAcc-cifar100}
    \end{subfigure}
    \hfill
    \begin{subfigure}{0.27\textwidth}
        \centering
        \includegraphics[width=\columnwidth]{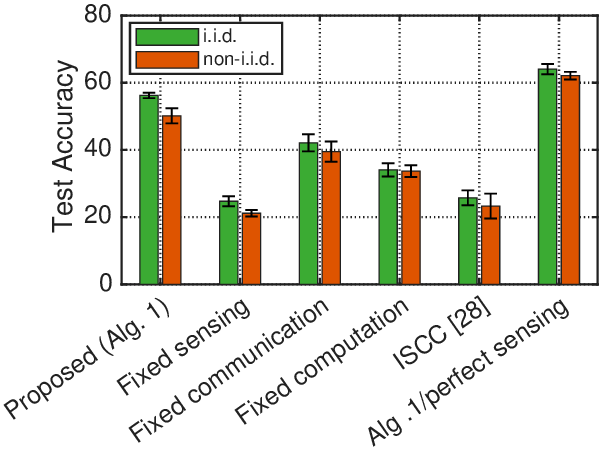}
        \caption{Final test accuracy, $T=1000$}
        \label{fig:TestAcc-final}
    \end{subfigure}

    \caption{Simulation results with the CIFAR-100 dataset. (a) Training loss for the proposed scheme and baselines with i.i.d. partitioning. (b) Test accuracy for the proposed scheme and baselines with non-i.i.d. partitioning. (c) Final test accuracy after $T=1000$ communication rounds. The shaded regions and error bars indicate one standard deviation over three seeds.}
    \label{fig:cifar100-results} 
      \vspace{-1em}
\end{figure*}


Fig.~\ref{fig:cifar100-results}(a) shows the training loss versus communication rounds for the proposed scheme and benchmark schemes on the CIFAR-100 dataset under i.i.d.\ data partitioning, with $T=1000$ rounds. The proposed joint design achieves the fastest loss reduction and the lowest steady-state training loss among all practical schemes, confirming its effectiveness on the more challenging CIFAR-100 task. The baselines exhibit higher loss because each optimizes only part of the sensing--computation--communication pipeline.

Fig.~\ref{fig:cifar100-results}(b) shows the test accuracy versus communication rounds on CIFAR-100 under non-i.i.d.\ data partitioning, where the proposed scheme consistently outperforms all practical baselines by jointly adapting sensing, computation, and communication resources. Fig.~\ref{fig:cifar100-results}(c) compares the final test accuracy after $T=1000$ rounds under both i.i.d.\ and non-i.i.d.\ data distributions. The proposed scheme achieves the highest final accuracy among all practical schemes in both cases, while the moderate gap between i.i.d.\ and non-i.i.d.\ results reflects the impact of data heterogeneity. For instance, the proposed scheme achieves approximately $12\%$ higher accuracy than the closest practical benchmark under non-i.i.d.\ partitioning.


\section{Conclusion}
\label{sec:conclusion}
In this paper, we investigate a unified wireless edge-learning framework that jointly optimizes sensing, computation, and communication in federated learning. Unlike formulations that assume data is already available at the devices, the framework models the sensing process explicitly and captures the coupled effects of sensing quality, local computation, and OTA aggregation under imperfect CSI on the learning trajectory. We derive a non-convex convergence bound that exposes these effects and use it to formulate a joint resource allocation problem balancing learning quality, energy, and latency, which we solve through a low-complexity alternating optimization method. The results highlight sensing as a fundamental and controllable component of wireless edge learning, with a significant impact on learning performance, and provide a principled basis for the joint design of sensing, computation, and communication under noisy multimodal sensing and imperfect CSI in future 6G edge-intelligence systems.

\appendices
 
\section{Proof of Lemma~\ref{lem:drift}}
\label{lemma:drift-proof}
 
By definition, $\bgamma^c_t = \sum_{n=1}^{N} \alpha_{n,t} \Delta\btheta^{\rm clean}_{n,t} + \eta\tau\nabla F_t(\btheta_t)$,
where $\alpha_{n,t} = D_{n,t}/D_t$ and $\Delta\btheta^{\rm clean}_{n,t}$ is the result of $\tau$ noise-free mini-batch SGD steps on $F_{n,t}$ starting at $\btheta_t$. Taking conditional expectation over the mini-batch sampling at device $n$ replaces the stochastic gradient by the empirical gradient:
\begin{align}
\mathbb{E}[\Delta\btheta^{\rm clean}_{n,t} \,|\, \btheta_t] = - \eta \sum_{i=1}^{\tau} \mathbb{E}\big[\nabla F_{n,t}(\btheta^{(i)}_{n,t}) \,\big|\, \btheta_t \big],
\end{align}
where $\btheta^{(i)}_{n,t}$ denotes the noise-free local iterate at step $i$ with $\btheta^{(1)}_{n,t} = \btheta_t$. Substituting into the definition of $\bgamma^c_t$ and rearranging,
\begin{align}
\mathbb{E}[\bgamma^c_t \,|\, \btheta_t] &= -\eta\sum_{n=1}^N \alpha_{n,t} \sum_{i=1}^{\tau} \mathbb{E}\big[\nabla F_{n,t}(\btheta^{(i)}_{n,t})\,\big|\,\btheta_t\big] + \eta\tau\nabla F_t(\btheta_t) \notag\\
&= \eta \sum_{i=1}^{\tau} \Big( \nabla F_t(\btheta_t) - \sum_{n=1}^N \alpha_{n,t}\, \mathbb{E}\big[\nabla F_{n,t}(\btheta^{(i)}_{n,t})\,\big|\,\btheta_t\big] \Big). \label{eq:drift-step1}
\end{align}
Adding and subtracting $\sum_n \alpha_{n,t}\nabla F_{n,t}(\btheta_t) = \nabla F_t(\btheta_t)$ inside each summand of~\eqref{eq:drift-step1},
\begin{align}
&\mathbb{E}[\bgamma^c_t \,|\, \btheta_t] \notag\\
&= \eta \sum_{i=1}^{\tau} \sum_{n=1}^N \alpha_{n,t} \Big( \nabla F_{n,t}(\btheta_t) - \mathbb{E}\big[\nabla F_{n,t}(\btheta^{(i)}_{n,t})\,\big|\,\btheta_t\big] \Big).
\end{align}
Taking norms, applying the triangle inequality, and using the $L$-smoothness of $F_{n,t}$ (A1),
\begin{align}
\big\|\mathbb{E}[\bgamma^c_t \,|\, \btheta_t]\big\| &\le \eta \sum_{i=1}^{\tau} \sum_{n=1}^N \alpha_{n,t} L\, \mathbb{E}\big[\big\|\btheta_t - \btheta^{(i)}_{n,t}\big\|\,\big|\,\btheta_t\big]. \label{eq:drift-step2}
\end{align}
The local-iterate drift from the noise-free recursion is bounded by induction. At step $j$, $\btheta^{(j+1)}_{n,t}=\btheta^{(j)}_{n,t}-\eta\nabla F_{n,t}(\btheta^{(j)}_{n,t};\bxi^{(j)}_{n,t})$, so
\begin{align}
\big\|\btheta_t-\btheta^{(j+1)}_{n,t}\big\| &\le \big\|\btheta_t-\btheta^{(j)}_{n,t}\big\|+\eta\big\|\nabla F_{n,t}(\btheta^{(j)}_{n,t};\bxi^{(j)}_{n,t})\big\| \notag\\
&\le \big\|\btheta_t-\btheta^{(j)}_{n,t}\big\|+\eta G,
\end{align}
where the last step uses the bounded gradient norm in A3. Unrolling the recursion from $\btheta^{(1)}_{n,t} = \btheta_t$ gives $\|\btheta_t - \btheta^{(i)}_{n,t}\| \le (i-1)\eta G \le \eta\tau G$ for $i \le \tau$. Substituting into~\eqref{eq:drift-step2} and squaring,
\begin{align}
\big\|\mathbb{E}[\bgamma^c_t \,|\, \btheta_t]\big\|^2 \le \big(\eta \cdot \tau \cdot L \cdot \eta\tau G\big)^2 = L^2\eta^4\tau^4 G^2,
\end{align}
which is~\eqref{eq:drift-bound}. \hfill$\blacksquare$

\section{Proof of Lemma~\ref{lem:variance}}
\label{lemma:variance-proof} 
Let $\bm{r}^{(i)}_{n,t} \triangleq \nabla F_{n,t}(\btheta^{(i)}_{n,t}; \bxi^{(i)}_{n,t}) - \mathbb{E}[\nabla F_{n,t}(\btheta^{(i)}_{n,t}; \bxi^{(i)}_{n,t})\,|\,\btheta^{(i)}_{n,t}]$ denote the mini-batch noise at local step $i$ on device $n$, conditional on the iterate $\btheta^{(i)}_{n,t}$. The sequence $\{\bm{r}^{(i)}_{n,t}\}_{i=1}^{\tau}$ forms a martingale-difference sequence conditioned on the local iterates. By A4 and Jensen's inequality, conditioning on $\btheta^{(i)}_{n,t}$,
\begin{align}
\mathbb{E}\big[\|\bm{r}^{(i)}_{n,t}\|^2\,\big|\,\btheta^{(i)}_{n,t}\big]\le \frac{\sigma_B^2}{B}+\frac{\sigma_{\rm pop}^2}{D_{n,t}}\triangleq V_{n,t}.
\end{align}
Since the local SGD mapping is $(1+\eta L)$-Lipschitz under A1, the conditional variance propagates as
\begin{align}
\mathbb{E}\big\|&\btheta^{(i+1)}_{n,t}-\mathbb{E}[\btheta^{(i+1)}_{n,t}\,|\,\btheta_t]\big\|^2 \notag \\ 
&\le (1+\eta L)^2\mathbb{E}\big\|\btheta^{(i)}_{n,t}-\mathbb{E}[\btheta^{(i)}_{n,t}\,|\,\btheta_t]\big\|^2+\eta^2V_{n,t}.
\end{align}
Since $\btheta^{(1)}_{n,t}=\btheta_t$, unrolling the recursion gives
\begin{align}
\mathbb{E}\big\|\Delta\btheta^{\rm clean}_{n,t}-\mathbb{E}[\Delta\btheta^{\rm clean}_{n,t}\,|\,\btheta_t]\big\|^2 & \le \eta^2V_{n,t}\sum_{i=0}^{\tau-1}(1+\eta L)^{2i} \notag \\
&\le 3\eta^2\tau V_{n,t},
\label{eq:variance-step2}
\end{align}
where the last inequality follows from $(1+\eta L)^{2\tau}\le e^{2\eta L\tau}\le e<3$ under $\eta\tau L\le 1/2$.

Aggregating across devices using the convex-combination structure and Jensen's inequality,
\begin{align}
\mathbb{E}\big\|\bgamma^c_t&-\mathbb{E}[\bgamma^c_t\,|\,\btheta_t]\big\|^2 \notag \\
& \le \mathbb{E}\Big\|\sum_n \alpha_{n,t}\big(\Delta\btheta^{\rm clean}_{n,t}-\mathbb{E}[\Delta\btheta^{\rm clean}_{n,t}\,|\,\btheta_t]\big)\Big\|^2 \notag\\
&\le \sum_n \alpha_{n,t}\mathbb{E}\big\|\Delta\btheta^{\rm clean}_{n,t}-\mathbb{E}[\Delta\btheta^{\rm clean}_{n,t}\,|\,\btheta_t]\big\|^2 \notag\\
&\le 3\eta^2\tau\sum_n \alpha_{n,t}V_{n,t}.
\end{align}
Substituting $V_{n,t}$ yields~\eqref{eq:variance-bound}. \hfill$\blacksquare$

\section{Proof of Theorem~\ref{thm:nonconvex}}
\label{app:proof-nonconvex}
 
By the $L$-smoothness of $F_t$ (A1),
\begin{align}
    F_t(\btheta_{t+1})
    &\le
    F_t(\btheta_t)
    + \nabla F_t(\btheta_t)^\top (\btheta_{t+1} - \btheta_t)
    \notag\\
    &\quad
    + \frac{L}{2}\|\btheta_{t+1} - \btheta_t\|^2.
    \label{eq:proof-smoothness}
\end{align}
Substituting the perturbed update $\btheta_{t+1} - \btheta_t = -\eta\tau\nabla F_t(\btheta_t) + \bzeta_t$ from~\eqref{eq:bs-update-perturbed},
\begin{align}
    F_t(\btheta_{t+1})
    &\le
    F_t(\btheta_t)
    - \eta\tau\|\nabla F_t(\btheta_t)\|^2
    + \nabla F_t(\btheta_t)^\top \bzeta_t
    \notag\\
    &\quad
    + \frac{L}{2}\big\|-\eta\tau\nabla F_t(\btheta_t) + \bzeta_t\big\|^2.
    \label{eq:proof-substitute}
\end{align}
Expanding the squared norm,
\begin{align}
    \big\|-\eta\tau\nabla F_t(\btheta_t) + \bzeta_t\big\|^2
    &=
    \eta^2\tau^2\|\nabla F_t(\btheta_t)\|^2
    \notag\\
    &\quad
    - 2\eta\tau\nabla F_t(\btheta_t)^\top\bzeta_t
    + \|\bzeta_t\|^2,
\end{align}
and collecting terms in~\eqref{eq:proof-substitute},
\begin{align}
    F_t(\btheta_{t+1})
    &\le
    F_t(\btheta_t)
    - \eta\tau\Big(1 - \frac{L\eta\tau}{2}\Big)\|\nabla F_t(\btheta_t)\|^2
    \notag\\
    &\quad
    + (1 - L\eta\tau)\nabla F_t(\btheta_t)^\top\bzeta_t
    + \frac{L}{2}\|\bzeta_t\|^2.
    \label{eq:proof-collected}
\end{align}
We take the conditional expectation given $\btheta_t$. The perturbation splits as $\bzeta_t = \bar\bzeta_t + \bzeta_t^0$ with $\mathbb{E}[\bzeta_t^0 \,|\, \btheta_t] = \bzero$, where $\bar\bzeta_t = \mathbb{E}[\bzeta_t \,|\, \btheta_t]$ collects the client-drift and OTA-misalignment biases. Hence
\begin{align}
    \mathbb{E}\big[\nabla F_t(\btheta_t)^\top\bzeta_t \,\big|\, \btheta_t\big]
    = \nabla F_t(\btheta_t)^\top\bar\bzeta_t.
\end{align}
For $\eta\tau L \le 1/2$ we have $0 \le 1 - L\eta\tau \le 1$. Applying Cauchy--Schwarz and Young's inequality with weight $\eta\tau/2$,
\begin{align}
    (1 - L\eta\tau)\nabla F_t(\btheta_t)^\top\bar\bzeta_t
    \le \frac{\eta\tau}{4}\|\nabla F_t(\btheta_t)\|^2 + \frac{1}{\eta\tau}\|\bar\bzeta_t\|^2.
    \label{eq:proof-young}
\end{align}
Since $\eta\tau(1 - L\eta\tau/2) \ge \eta\tau/2$ for $\eta\tau L \le 1/2$, the net coefficient of $\|\nabla F_t(\btheta_t)\|^2$ after subtracting the $\eta\tau/4$ term in~\eqref{eq:proof-young} is at most $-\eta\tau/4$. Taking total expectation,
\begin{align}
    \mathbb{E}[F_t(\btheta_{t+1})]
    &\le
    \mathbb{E}[F_t(\btheta_t)]
    - \frac{\eta\tau}{4}\mathbb{E}\|\nabla F_t(\btheta_t)\|^2
    \notag\\
    &\quad
    + \frac{1}{\eta\tau}\mathbb{E}\|\bar\bzeta_t\|^2
    + \frac{L}{2}\mathbb{E}\|\bzeta_t\|^2.
    \label{eq:proof-cond-expectation}
\end{align}
The bias norm satisfies $\|\bar\bzeta_t\|^2 \le \mathbb{E}\|\bzeta_t\|^2 \le \Psi_t$ by Jensen's inequality and Lemma~\ref{lem:total-perturbation}, and $\mathbb{E}\|\bzeta_t\|^2 \le \Psi_t$ directly. Substituting both,
\begin{align}
    \mathbb{E}[F_t(\btheta_{t+1})]
    \le
    &\,\,\mathbb{E}[F_t(\btheta_t)]
    - \frac{\eta\tau}{4}\mathbb{E}\|\nabla F_t(\btheta_t)\|^2 \notag \\
    &
    + \Big(\frac{1}{\eta\tau} + \frac{L}{2}\Big)\Psi_t.
    \label{eq:proof-descent}
\end{align}
To telescope across rounds with the time-varying objective, we bridge $F_t$ and $F_{t+1}$. By the definition of the drift,
\begin{align}
    \mathbb{E}[F_{t+1}(\btheta_{t+1})] = \mathbb{E}[F_t(\btheta_{t+1})] + \delta_t.
\end{align}
Substituting into~\eqref{eq:proof-descent},
\begin{align}
    \mathbb{E}[F_{t+1}(\btheta_{t+1})]
    &\le
    \mathbb{E}[F_t(\btheta_t)]
    - \frac{\eta\tau}{4}\mathbb{E}\|\nabla F_t(\btheta_t)\|^2
    \notag\\
    &\quad
    + \Big(\frac{1}{\eta\tau} + \frac{L}{2}\Big)\Psi_t + \delta_t.
    \label{eq:proof-bridged}
\end{align}
Rearranging and summing over $t = 0, \ldots, T-1$, the $\mathbb{E}[F_t(\btheta_t)]$ terms telescope:
\begin{align}
    \frac{\eta\tau}{4}\sum_{t=0}^{T-1}\mathbb{E}\|\nabla F_t(\btheta_t)\|^2
    &\le
    \mathbb{E}[F_0(\btheta_0)] - \mathbb{E}[F_T(\btheta_T)]
    \notag\\
    &
    + \Big(\frac{1}{\eta\tau} + \frac{L}{2}\Big)\sum_{t=0}^{T-1}\Psi_t
    + \sum_{t=0}^{T-1}\delta_t.
\end{align}
Using $\mathbb{E}[F_T(\btheta_T)] \ge F_{\inf}$ from A2 and dividing by $\eta\tau T/4$,
\begin{align}
    \frac{1}{T}\sum_{t=0}^{T-1}\,&\mathbb{E}\|\nabla F_t(\btheta_t)\|^2
    \le
    \frac{4\big(\mathbb{E}[F_0(\btheta_0)] - F_{\inf}\big)}{\eta\tau T}
    \notag\\
    &\quad
    + \Big(\frac{4}{\eta^2\tau^2} + \frac{2L}{\eta\tau}\Big)\frac{1}{T}\sum_{t=0}^{T-1}\Psi_t
    + \frac{4}{\eta\tau T}\sum_{t=0}^{T-1}\delta_t.
\end{align}
For $\eta\tau L \le 1/2$, this is the bound in~\eqref{eq:nonconvex-main-bound}, with the penalty coefficient $\frac{4}{\eta^2\tau^2} + \frac{2L}{\eta\tau} = \mathcal{O}(1/(\eta\tau)^2)$. This completes the proof.

\bibliographystyle{IEEEtran}
\bibliography{IEEEabrv,ref}

\balance

\end{document}